\documentclass[11pt,english]{article}
\usepackage{lmodern}

\usepackage[T1]{fontenc}
\usepackage[latin9]{inputenc}
\usepackage{geometry}
\usepackage{float}
\usepackage{algorithmic}
\usepackage{mathtools}
\usepackage{dsfont}
\usepackage{amsmath}
\usepackage{amsthm}
\usepackage{amssymb}

\makeatletter

\floatstyle{ruled}
\newfloat{algorithm}{tbp}{loa}
\providecommand{\algorithmname}{Algorithm}
\floatname{algorithm}{\protect\algorithmname}

\numberwithin{figure}{section}
\numberwithin{equation}{section}
\theoremstyle{plain}
\newtheorem{thm}{\protect\theoremname}
\theoremstyle{plain}
\newtheorem{lem}[thm]{\protect\lemmaname}
\theoremstyle{plain}
\newtheorem{question}[thm]{\protect\questionname}
\theoremstyle{definition}
\newtheorem{defn}[thm]{\protect\definitionname}
\theoremstyle{plain}
\newtheorem{prop}[thm]{\protect\propositionname}
\theoremstyle{plain}
\newtheorem{cor}[thm]{\protect\corollaryname}
\theoremstyle{definition}
\newtheorem{rmk}[thm]{\protect\remarkname}

\usepackage{comment,url,algorithm,algorithmic,graphicx,subcaption,relsize}
\usepackage{amssymb,amsfonts,amsmath,amsthm,amscd,dsfont,mathrsfs,mathtools,microtype,nicefrac,pifont}
\usepackage{float,psfrag,epsfig,color,url,hyperref}
\hypersetup{
  colorlinks,
  linkcolor={red!50!black},
  citecolor={blue!50!black},
  urlcolor={blue!80!black}
}
\usepackage{upgreek}
\usepackage[dvipsnames]{xcolor}
\usepackage{epstopdf,bbm,mathtools,enumitem}
\usepackage[toc,page]{appendix}
\usepackage{dsfont}

\usepackage[mathscr]{euscript}
\allowdisplaybreaks
\usepackage{custom}

\makeatletter
\renewcommand{\paragraph}{%
  \@startsection{paragraph}{4}%
  {\z@}{1.25ex \@plus 1ex \@minus .2ex}{-1em}%
  {\normalfont\normalsize\bfseries}%
}
\makeatother

\makeatother

\usepackage{babel}
\providecommand{\definitionname}{Definition}
\providecommand{\lemmaname}{Lemma}
\providecommand{\propositionname}{Proposition}
\providecommand{\remarkname}{Remark}
\providecommand{\questionname}{Question}
\providecommand{\theoremname}{Theorem}
\providecommand{\corollaryname}{Corollary}

\begin{document}
\def\balign#1\ealign{\begin{align}#1\end{align}}
\def\baligns#1\ealigns{\begin{align*}#1\end{align*}}
\def\balignat#1\ealign{\begin{alignat}#1\end{alignat}}
\def\balignats#1\ealigns{\begin{alignat*}#1\end{alignat*}}
\def\bitemize#1\eitemize{\begin{itemize}#1\end{itemize}}
\def\benumerate#1\eenumerate{\begin{enumerate}#1\end{enumerate}}

\newenvironment{talign*}
 {\let\displaystyle\textstyle\csname align*\endcsname}
 {\endalign}
\newenvironment{talign}
 {\let\displaystyle\textstyle\csname align\endcsname}
 {\endalign}

\def\balignst#1\ealignst{\begin{talign*}#1\end{talign*}}
\def\balignt#1\ealignt{\begin{talign}#1\end{talign}}

\let\originalleft\left
\let\originalright\right
\renewcommand{\left}{\mathopen{}\mathclose\bgroup\originalleft}
\renewcommand{\right}{\aftergroup\egroup\originalright}

\def\Gronwall{Gr\"onwall\xspace}
\def\Holder{H\"older\xspace}
\def\Ito{It\^o\xspace}
\def\Nystrom{Nystr\"om\xspace}
\def\Schatten{Sch\"atten\xspace}
\def\Matern{Mat\'ern\xspace}

\def\tinycitep*#1{{\tiny\citep*{#1}}}
\def\tinycitealt*#1{{\tiny\citealt*{#1}}}
\def\tinycite*#1{{\tiny\cite*{#1}}}
\def\smallcitep*#1{{\scriptsize\citep*{#1}}}
\def\smallcitealt*#1{{\scriptsize\citealt*{#1}}}
\def\smallcite*#1{{\scriptsize\cite*{#1}}}

\def\blue#1{\textcolor{blue}{{#1}}}
\def\green#1{\textcolor{green}{{#1}}}
\def\orange#1{\textcolor{orange}{{#1}}}
\def\purple#1{\textcolor{purple}{{#1}}}
\def\red#1{\textcolor{red}{{#1}}}
\def\teal#1{\textcolor{teal}{{#1}}}

\def\mbi#1{\boldsymbol{#1}} 
\def\mbf#1{\mathbf{#1}}
\def\mrm#1{\mathrm{#1}}
\def\tbf#1{\textbf{#1}}
\def\tsc#1{\textsc{#1}}

\def\mbiA{\mbi{A}}
\def\mbiB{\mbi{B}}
\def\mbiC{\mbi{C}}
\def\mbiDelta{\mbi{\Delta}}
\def\mbif{\mbi{f}}
\def\mbiF{\mbi{F}}
\def\mbih{\mbi{g}}
\def\mbiG{\mbi{G}}
\def\mbih{\mbi{h}}
\def\mbiH{\mbi{H}}
\def\mbiI{\mbi{I}}
\def\mbim{\mbi{m}}
\def\mbiP{\mbi{P}}
\def\mbiQ{\mbi{Q}}
\def\mbiR{\mbi{R}}
\def\mbiv{\mbi{v}}
\def\mbiV{\mbi{V}}
\def\mbiW{\mbi{W}}
\def\mbiX{\mbi{X}}
\def\mbiY{\mbi{Y}}
\def\mbiZ{\mbi{Z}}

\def\textsum{{\textstyle\sum}} 
\def\textprod{{\textstyle\prod}} 
\def\textbigcap{{\textstyle\bigcap}} 
\def\textbigcup{{\textstyle\bigcup}} 

\def\reals{\mathbb{R}} 
\def\integers{\mathbb{Z}} 
\def\rationals{\mathbb{Q}} 
\def\naturals{\mathbb{N}} 
\def\complex{\mathbb{C}} 

\def\what#1{\widehat{#1}}

\def\twovec#1#2{\left[\begin{array}{c}{#1} \\ {#2}\end{array}\right]}
\def\threevec#1#2#3{\left[\begin{array}{c}{#1} \\ {#2} \\ {#3} \end{array}\right]}
\def\nvec#1#2#3{\left[\begin{array}{c}{#1} \\ {#2} \\ \vdots \\ {#3}\end{array}\right]} 

\def\maxeig#1{\lambda_{\mathrm{max}}\left({#1}\right)}
\def\mineig#1{\lambda_{\mathrm{min}}\left({#1}\right)}

\def\Re{\operatorname{Re}} 
\def\indic#1{\mbb{I}\left[{#1}\right]} 
\def\logarg#1{\log\left({#1}\right)} 
\def\polylog{\operatorname{polylog}}
\def\maxarg#1{\max\left({#1}\right)} 
\def\minarg#1{\min\left({#1}\right)} 
\def\Earg#1{\E\left[{#1}\right]}
\def\Esub#1{\E_{#1}}
\def\Esubarg#1#2{\E_{#1}\left[{#2}\right]}
\def\bigO#1{\mathcal{O}\left(#1\right)} 
\def\littleO#1{o(#1)} 
\def\P{\mbb{P}} 
\def\Parg#1{\P\left({#1}\right)}
\def\Psubarg#1#2{\P_{#1}\left[{#2}\right]}
\def\Trarg#1{\Tr\left[{#1}\right]} 
\def\trarg#1{\tr\left[{#1}\right]} 
\def\Var{\mrm{Var}} 
\def\Vararg#1{\Var\left[{#1}\right]}
\def\Varsubarg#1#2{\Var_{#1}\left[{#2}\right]}
\def\Cov{\mrm{Cov}} 
\def\Covarg#1{\Cov\left[{#1}\right]}
\def\Covsubarg#1#2{\Cov_{#1}\left[{#2}\right]}
\def\Corr{\mrm{Corr}} 
\def\Corrarg#1{\Corr\left[{#1}\right]}
\def\Corrsubarg#1#2{\Corr_{#1}\left[{#2}\right]}
\newcommand{\info}[3][{}]{\mathbb{I}_{#1}\left({#2};{#3}\right)} 
\newcommand{\staticexp}[1]{\operatorname{exp}(#1)} 
\newcommand{\loglihood}[0]{\mathcal{L}} 


\providecommand{\arccos}{\mathop\mathrm{arccos}}
\providecommand{\dom}{\mathop\mathrm{dom}}
\providecommand{\diag}{\mathop\mathrm{diag}}
\providecommand{\tr}{\mathop\mathrm{tr}}
\providecommand{\card}{\mathop\mathrm{card}}
\providecommand{\sign}{\mathop\mathrm{sign}}
\providecommand{\conv}{\mathop\mathrm{conv}} 
\def\rank#1{\mathrm{rank}({#1})}
\def\supp#1{\mathrm{supp}({#1})}

\providecommand{\minimize}{\mathop\mathrm{minimize}}
\providecommand{\maximize}{\mathop\mathrm{maximize}}
\providecommand{\subjectto}{\mathop\mathrm{subject\;to}}

\def\openright#1#2{\left[{#1}, {#2}\right)}

\ifdefined\nonewproofenvironments\else
\ifdefined\ispres\else
\newenvironment{proof-of-theorem}[1][{}]{\noindent\textbf{Proof of Theorem {#1}}
   \hspace*{1em}}{\qed\\}
 
\fi
\makeatletter
\@addtoreset{equation}{section}
\makeatother
\def\theequation{\thesection.\arabic{equation}}

\newcommand{\cmark}{\ding{51}}

\newcommand{\xmark}{\ding{55}}

\newcommand{\eq}[1]{\begin{align}#1\end{align}}
\newcommand{\eqn}[1]{\begin{align*}#1\end{align*}}
\renewcommand{\Pr}[1]{\mathbb{P}\left( #1 \right)}
\newcommand{\Ex}[1]{\mathbb{E}\left[#1\right]}

\newcommand{\matt}[1]{{\textcolor{Maroon}{[Matt: #1]}}}
\newcommand{\kook}[1]{{\textcolor{blue}{[Kook: #1]}}}
\definecolor{OliveGreen}{rgb}{0,0.6,0}
\newcommand{\sv}[1]{{\textcolor{OliveGreen}{[Santosh: #1]}}}

\global\long\def\on#1{\operatorname{#1}}%

\global\long\def\bw{\mathsf{Ball\ walk}}%
\global\long\def\sw{\mathsf{Speedy\ walk}}%
\global\long\def\har{\mathsf{Hit\text{-}and\text{-}Run}}%
\global\long\def\ino{\mathsf{In\text{-}and\text{-}Out}}%
\global\long\def\gw{\mathsf{Gaussian\ walk}}%
\global\long\def\ps{\mathsf{Proximal\ sampler}}%
\global\long\def\gc{\mathsf{Gaussian\ cooling}}%
\global\long\def\sps{\mathsf{Proximal\ Gaussian \ cooling}}%
\global\long\def\psunif{\mathsf{Prox\text{-}uniform}}%
\global\long\def\psgauss{\mathsf{Prox\text{-}Gaussian}}%

\global\long\def\dw{\mathsf{Dikin\ walk}}%
\global\long\def\rhmc{\mathsf{Riemannian\ Hamiltonian\ Monte\ Carlo}}%
\global\long\def\rlmc{\mathsf{Riemannian\ Langevin\ Monte\ Carlo}}%
\global\long\def\ml{\mathsf{Mirror\ Langevin}}%

\global\long\def\O{\mathcal{O}}%
\global\long\def\Otilde{\widetilde{\mathcal{O}}}%
\global\long\def\Omtilde{\widetilde{\Omega}}%

\global\long\def\E{\mathbb{E}}%
\global\long\def\Z{\mathbb{Z}}%
\global\long\def\P{\mathbb{P}}%
\global\long\def\N{\mathbb{N}}%

\global\long\def\K{\mathcal{K}}%

\global\long\def\R{\mathbb{R}}%
\global\long\def\Rd{\mathbb{R}^{d}}%
\global\long\def\Rdd{\mathbb{R}^{d\times d}}%
\global\long\def\Rn{\mathbb{R}^{n}}%
\global\long\def\Rnn{\mathbb{R}^{n\times n}}%

\global\long\def\psd{\mathbb{S}_{+}^{d}}%
\global\long\def\pd{\mathbb{S}_{++}^{d}}%

\global\long\def\defeq{\stackrel{\mathrm{{\scriptscriptstyle def}}}{=}}%

\global\long\def\cpi{C_{\mathsf{PI}}}%
\global\long\def\clsi{C_{\mathsf{LSI}}}%

\global\long\def\veps{\varepsilon}%
\global\long\def\lda{\lambda}%
\global\long\def\vphi{\varphi}%

\global\long\def\half{\frac{1}{2}}%
\global\long\def\nhalf{\nicefrac{1}{2}}%
\global\long\def\texthalf{{\textstyle \frac{1}{2}}}%

\global\long\def\ind{\mathds{1}}%
\global\long\def\op{\mathsf{op}}%

\global\long\def\chooses#1#2{_{#1}C_{#2}}%

\global\long\def\vol{\on{vol}}%

\global\long\def\law{\on{law}}%

\global\long\def\tr{\on{tr}}%

\global\long\def\diag{\on{diag}}%

\global\long\def\Diag{\on{Diag}}%

\global\long\def\inter{\on{int}}%

\global\long\def\esssup{\on{ess\,sup}}%

\global\long\def\e{\mathrm{e}}%

\global\long\def\id{\mathrm{id}}%

\global\long\def\spanning{\on{span}}%

\global\long\def\rows{\on{row}}%

\global\long\def\cols{\on{col}}%

\global\long\def\rank{\on{rank}}%

\global\long\def\T{\mathsf{T}}%

\global\long\def\bs#1{\boldsymbol{#1}}%

\global\long\def\eu#1{\EuScript{#1}}%

\global\long\def\mb#1{\mathbf{#1}}%

\global\long\def\mbb#1{\mathbb{#1}}%

\global\long\def\mc#1{\mathcal{#1}}%

\global\long\def\mf#1{\mathfrak{#1}}%

\global\long\def\ms#1{\mathscr{#1}}%

\global\long\def\mss#1{\mathsf{#1}}%

\global\long\def\msf#1{\mathsf{#1}}%

\global\long\def\textint{{\textstyle \int}}%
\global\long\def\Dd{\mathrm{D}}%
\global\long\def\D{\mathrm{d}}%
\global\long\def\grad{\nabla}%
 
\global\long\def\hess{\nabla^{2}}%
 
\global\long\def\lapl{\triangle}%
 
\global\long\def\deriv#1#2{\frac{\D#1}{\D#2}}%
 
\global\long\def\pderiv#1#2{\frac{\partial#1}{\partial#2}}%
 
\global\long\def\de{\partial}%
\global\long\def\lagrange{\mathcal{L}}%
\global\long\def\Div{\on{div}}%

\global\long\def\Gsn{\mathcal{N}}%
 
\global\long\def\BeP{\textnormal{BeP}}%
 
\global\long\def\Ber{\textnormal{Ber}}%
 
\global\long\def\Bern{\textnormal{Bern}}%
 
\global\long\def\Bet{\textnormal{Beta}}%
 
\global\long\def\Beta{\textnormal{Beta}}%
 
\global\long\def\Bin{\textnormal{Bin}}%
 
\global\long\def\BP{\textnormal{BP}}%
 
\global\long\def\Dir{\textnormal{Dir}}%
 
\global\long\def\DP{\textnormal{DP}}%
 
\global\long\def\Expo{\textnormal{Expo}}%
 
\global\long\def\Gam{\textnormal{Gamma}}%
 
\global\long\def\GEM{\textnormal{GEM}}%
 
\global\long\def\HypGeo{\textnormal{HypGeo}}%
 
\global\long\def\Mult{\textnormal{Mult}}%
 
\global\long\def\NegMult{\textnormal{NegMult}}%
 
\global\long\def\Poi{\textnormal{Poi}}%
 
\global\long\def\Pois{\textnormal{Pois}}%
 
\global\long\def\Unif{\textnormal{Unif}}%

\global\long\def\bpar#1{\bigl(#1\bigr)}%
\global\long\def\Bpar#1{\Bigl(#1\Bigr)}%

\global\long\def\snorm#1{\|#1\|}%
\global\long\def\bnorm#1{\bigl\Vert#1\bigr\Vert}%
\global\long\def\Bnorm#1{\Bigl\Vert#1\Bigr\Vert}%

\global\long\def\sbrack#1{[#1]}%
\global\long\def\bbrack#1{\bigl[#1\bigr]}%
\global\long\def\Bbrack#1{\Bigl[#1\Bigr]}%

\global\long\def\sbrace#1{\{#1\}}%
\global\long\def\bbrace#1{\bigl\{#1\bigr\}}%
\global\long\def\Bbrace#1{\Bigl\{#1\Bigr\}}%

\global\long\def\Abs#1{\left\lvert #1\right\rvert }%
\global\long\def\Par#1{\left(#1\right)}%
\global\long\def\Brack#1{\left[#1\right]}%
\global\long\def\Brace#1{\left\{  #1\right\}  }%

\global\long\def\inner#1{\langle#1\rangle}%
 
\global\long\def\binner#1#2{\left\langle {#1},{#2}\right\rangle }%

\global\long\def\norm#1{\|#1\|}%
\global\long\def\onenorm#1{\norm{#1}_{1}}%
\global\long\def\twonorm#1{\norm{#1}_{2}}%
\global\long\def\infnorm#1{\norm{#1}_{\infty}}%
\global\long\def\fronorm#1{\norm{#1}_{\text{F}}}%
\global\long\def\nucnorm#1{\norm{#1}_{*}}%
\global\long\def\staticnorm#1{\|#1\|}%
\global\long\def\statictwonorm#1{\staticnorm{#1}_{2}}%

\global\long\def\mmid{\mathbin{\|}}%

\global\long\def\otilde#1{\widetilde{\mc O}(#1)}%
\global\long\def\wtilde{\widetilde{W}}%
\global\long\def\wt#1{\widetilde{#1}}%

\global\long\def\KL{\msf{KL}}%
\global\long\def\dtv{d_{\textrm{\textup{TV}}}}%
\global\long\def\FI{\msf{FI}}%
\global\long\def\tv{\msf{TV}}%

\global\long\def\cov{\mathrm{Cov}}%
\global\long\def\var{\mathrm{Var}}%

\global\long\def\cred#1{\textcolor{red}{#1}}%
\global\long\def\cblue#1{\textcolor{blue}{#1}}%
\global\long\def\cgreen#1{\textcolor{green}{#1}}%
\global\long\def\ccyan#1{\textcolor{cyan}{#1}}%

\global\long\def\iff{\Leftrightarrow}%

\global\long\def\bmu{\bar{\mu}}
\global\long\def\bK{\bar{\mc K}}
 
\global\long\def\textfrac#1#2{{\textstyle \frac{#1}{#2}}}%

\title{R\'enyi-infinity constrained sampling with $d^3$ membership queries}
\author{
Yunbum Kook\thanks{
  School of Computer Science,
  Georgia Institute of Technology, \texttt{yb.kook@gatech.edu}
} \ \ \ \ \ \ \ \ \ \ \ \ \ \ \ \ \ 
Matthew S.\ Zhang\thanks{
  Dept. of Computer Science,
  University of Toronto, and Vector Institute, \texttt{matthew.zhang@mail.utoronto.ca}
}
}
\date{ }
\maketitle

\begin{abstract}
    Uniform sampling over a convex body is a fundamental algorithmic problem, yet the convergence in KL or R\'enyi divergence of most samplers remains poorly understood. 
    In this work, we propose a constrained proximal sampler, a principled and simple algorithm that possesses elegant convergence guarantees. Leveraging the uniform ergodicity of this sampler, we show that it converges in the R\'enyi-infinity divergence ($\EuScript R_\infty$) with no query complexity overhead when starting from a warm start. This is the strongest of commonly considered performance metrics, implying rates in $\{\EuScript R_q, \mathsf{KL}\}$ convergence as special cases.

    By applying this sampler within an annealing scheme, we propose an algorithm which can approximately sample $\varepsilon$-close to the uniform distribution on convex bodies in $\EuScript R_\infty$-divergence with $\widetilde{\mathcal{O}}(d^3\, \text{polylog} \frac{1}{\varepsilon})$ query complexity. This improves on all prior results in $\{\EuScript R_q, \mathsf{KL}\}$-divergences, without resorting to any algorithmic modifications or post-processing of the sample. It also matches the prior best known complexity in total variation distance.
\end{abstract}

\section{Introduction}\label{sec:intro}

Uniform sampling from convex bodies is a fundamental question in computer science. 
Its applications have spanned from differential privacy~\cite{mcsherry2007mechanism, hardt2010geometry,mironov2017renyi} to scientific computing~\cite{cousins2016practical,haraldsdottir2017chrr,kook2022sampling} and machine learning~\cite{bingham2019pyro,stan}.

The standard computational model assumes that, given a convex body $\K \subset \R^d$ which satisfies $B_1(0) \subseteq \K \subseteq B_D(0)$, the algorithm can access a \emph{membership oracle} \cite{grotschel2012geometric} which, given a point $x \in \R^d$, answers \emph{Yes} or \emph{No} to the query ``Is $x \in \K$?''.
This framework possesses two advantages: (i) it is the most general framework under which one can conduct analysis, subsuming other computational models as particular cases,
(ii) it has been thoroughly studied both in optimization and sampling.

Under this computational model, what is the complexity in terms of membership oracle queries to generate a random sample from a convex body $\K$ (i.e., from $\pi = \frac{1}{\vol(\K)}\ind_{\K}$)?
As many of the desired applications require the dimension $d$ to be large, it is typically impractical to generate samples exactly from the uniform distribution. Among other difficulties, this is due to the intractability of deterministically computing the normalizing factor $\vol(\K)$. Instead, one resorts to sampling from an approximate distribution which is $\varepsilon$-close to the uniform in an appropriate sense. Subsequent work has aimed to develop algorithms and analysis with optimal asymptotic dependence both on the dimension $d$ and error tolerance $\varepsilon$.

The seminal work of \cite{dyer1991random} proposed a randomized polynomial time algorithm with oracle complexity $\text{poly}(d,\log\frac{D}{\veps})$.
Since then, a proliferation of improvements to both the algorithm and analysis~\cite{lovasz1990mixing, lovasz1991compute, applegate1991sampling, lovasz1993random,kannan1997random,lovasz2006hit,cousins2018gaussian} have led to a refined understanding of the complexity of this problem. Such algorithms are typically instances of a random walk, which first samples a point $x_0 \sim \mu_0$ for a tractable initial measure $\mu_0$, and then produces $x_1, x_2, \ldots, x_{k_*}$ following a well-specified, possibly random procedure for some number of iterations $k_*$. Under mild assumptions, it is possible to make $\mu_{k_*}$, the distribution of $x_{k_*}$, $\varepsilon$-close to $\pi$ in some desired sense. These iterative schemes are called often Markov chain Monte Carlo (MCMC) methods.

The notion of closeness can take many different forms. Natural choices of metric\footnote{In this text, we will often refer to any of these divergences as a ``metric'' in the sense of performance metric, although only the total variation is a true metric on the space of probability measures.} such as the $\KL$ divergence and total variation ($\tv$) fit into the following hierarchy (see \S\ref{sec:prelim} for definitions):
\[
    \tv^2 \le \msf{KL}\text{ divergence} \leq \eu R_q\text{ divergence} \leq \eu R_\infty\text{ divergence}.
\]
While the complexity of uniformly sampling convex bodies was known to be $\Otilde(d^3\log^2\frac{1}{\veps})$ for $\veps$-$\tv$~\cite{cousins2018gaussian,jia2021reducing}, the best known rates in any other metric had much worse dependence on key parameters. Additionally, they relied on algorithmic modifications which deviated from practical implementations~\cite{pmlr-v65-brosse17a, gurbuzbalaban2022penalized, mangoubi2022sampling}. 
On the other hand, the literature on unconstrained sampling has provided stronger guarantees in $\KL$ or $\chi^2$ \emph{without} overhead in complexity, through refined analyses of Langevin-based algorithms~\cite{vempala2019rapid, chewi2021analysis, chen2022improved}.
Adapting those guarantees to the convex body setting is not trivial, since restricting the distribution to $\mc K$ places a strong constraint on the candidate algorithm. The discrepancies in complexity when going beyond $\tv$ leads us to a fundamental question on the complexity of uniform sampling in a stronger metric:
\begin{center}
\emph{Is it possible to achieve a query complexity of $\Otilde(d^3 \polylog \frac{1}{\varepsilon})$ for uniform sampling in $\{\KL, \eu R_q, \eu R_\infty\}$-divergence?}
\end{center}

In general, almost all algorithms take two phases: ($\msf{P1}$) \emph{warm-start generation} and ($\msf{P2}$) \emph{faster sampling under warm-start}. 
$\msf{P1}$ finds a tractable initial distribution $\mu_0$ towards the target $\pi$, whose closeness to $\pi$ is measured by the $\eu R_\infty$ divergence at initialization, $M = \sup \frac{\D \mu_0}{\D\pi} = \exp (\eu R_\infty(\mu_0 \mmid \pi))$. $\msf{P2}$ then samples from the target by leveraging the initial warmness. 
This approach is also taken by~\cite{cousins2018gaussian}, the best known uniform-sampling algorithm in $\tv$. However, it fails to go beyond $\tv$, as a warm-start generated by the algorithm holds only in a \emph{weak sense}, causing ``$\tv$-collapse'' (we return to this shortly). The resulting sampler can only have guarantees in $\tv$.

Ignoring $\msf{P1}$ for the moment, one hopes that, towards stronger guarantees, there should be a sampler which converges in $\eu R_q$ in $\msf{P2}$. 
Recently, \cite{kook2024inout} proposed an algorithm with query complexity $\Otilde(qMd^2 \norm{\Sigma}_{\op} \polylog \frac{1}{\varepsilon})$ for $\Sigma$ covariance of $\pi$, which seamlessly extends the best complexity of $\Otilde(Md^2\norm{\Sigma}_{\op} \polylog \frac{1}{\varepsilon})$ in $\tv$ that is achieved by $\bw$. 
However, this immediately raises a problem: 
it requires a $\eu R_\infty$ guarantee at initialization in order to achieve $\eu R_q$-convergence. 
One might want to sample from $\pi$ directly from some simple feasible start, such as a uniform distribution on the unit ball. Unfortunately, such choices potentially incur exponential warmness in dimension (i.e., $M\gtrsim \exp(d)$). This is where $\msf{P1}$ demonstrates its importance, allowing us to address this issue through an \emph{annealing scheme}.

\paragraph{Annealing algorithms}
The strategy of annealing for volume computation and uniform sampling appeared in~\cite{dyer1991random}.
It proceeds with a sequence of distributions $\{\mu_k \}_{k \leq k_*}$ which steadily evolve from $\mu_0 \propto \ind_{B_1(0)}$, the uniform distribution on a unit ball, towards $\mu_{k_*} = \pi$ the uniform distribution on the entire body. This is done by fattening $\mu_k$, steadily increasing its variance, until it is close enough to $\pi$ to be sampled directly with MCMC.\footnote{In the original paper, this is done by a sequence of uniform distributions, but this strategy has since been superseded by considering a sequence of truncated Gaussians with appropriate variance.} This general schema, wherein a tractable, well-understood distribution is gradually converted to a generic one, has been called \emph{annealing}.

This scheme has been refined by annealing via exponential distributions~\cite{lovasz2006simulated} and via Gaussians.
The latter, whose use of $\{\mc N(0, \sigma^2_k I_d)|_{\mc K}\}_{k \leq k_*}$ was pioneered by the $\gc$ algorithm~\cite{cousins2018gaussian}, turns out to be a robust choice for $\{\mu_k\}_{k \leq k_*}$ under an appropriate update schedule for the variance $\sigma_k^2$. 
These intermediate distributions are ``stepping stones'' towards the desired uniform distribution, in the sense that each distribution is $\O(1)$ close to the subsequent one in $\eu R_\infty$, and no more than $k_* = \Otilde(d)$ annealing stages are required in total.

$\gc$ can guarantee $M = \O(1)$ in a weak sense with the total query complexity $\Otilde(C^2 d^3 \polylog \frac{1}{\varepsilon})$ under a canonical set-up where $\K$ is \emph{well-rounded} (i.e., $\E_{X \sim \pi}[\norm{X}^2] \leq C^2 d$ for some constant $C$) \cite{kannan1997random,lovasz2006simulated,lovasz2006fast,cousins2018gaussian}.
Later work~\cite{jia2021reducing} claims that one can find an affine transformation which maps a skewed convex body to a well-rounded one with $C=\O(1)$, via a randomized algorithm with query complexity $\Otilde(d^3)$. 

\paragraph{TV-collapse}
Such annealing schemes are interwoven with an \emph{approximate} sampler, typically an MCMC subroutine such as $\bw$ and $\har$, requiring the sampler to approximately sample from evolving measures. 
Hence, the warmness for a next target $\mu_{k+1}$ is actually measured with respect to a distribution $\tilde{\mu}_{k}$, approximately close to $\mu_k$.
Standard constrained samplers like $\bw$ or $\har$ are only able to provide guarantees (between $\tilde{\mu}_{k}$ and $\mu_k$) in $\tv$ for truncated Gaussian. However, due to failure of the triangle inequality through $\tv$ and $\eu R_q$ in general, the warmness of $\tilde{\mu}_{k}$ with respect to $\mu_{k+1}$ cannot be bounded. Even worse, it is essential that the warm-start for these samplers is in some R\'enyi divergence, preferably the $\eu R_{\infty}$ divergence. As a result, the annealing schemes fail to carry a strong enough notion of warmness across the evolving measures. Instead, it passes a \emph{weakened form of warmness} to the last step via a coupling argument, causing any sampling guarantees based on this warmness to degrade to $\tv$ distance.

Given this challenge, we state a refinement of ($\msf{P1}$):
\begin{question}
\begin{centering}
    \emph{Can we generate a \textbf{true} warm-start $M=\O(1)$ with $\Otilde(d^3)$ queries?}
\end{centering}
\end{question}
The answer to this hinges upon the existence of a sampler which converges in $\eu R_\infty$. As the triangle inequality holds for $\eu R_\infty$, this would enable us to anneal the distribution \emph{without compromising} $\eu R_\infty$-closeness throughout! Examining this sends us back to a more specific form of ($\msf{P2}$):
\begin{question}\label{q:warm-question}
\begin{centering}
{\em Given $M=\O(1)$ and uniform or truncated Gaussian target, is there a sampler with $\eu R_\infty$ guarantees \textbf{without} any overhead in query complexity?}
\end{centering}
\end{question}

\paragraph{R\'enyi divergence: theoretical gap and privacy}
From our earlier discussion, we cannot help but circle back to the study of R\'enyi divergence in constrained sampling. This is theoretically important in its own right, since it leads us to examine a current gap between unconstrained and constrained sampling.
More specifically, the improvement in metric from $\tv$ to $\{\KL, \eu R_q\}$ nicely parallels the work in unconstrained sampling, where convergence in weaker metrics (Wasserstein, $\tv$)~\cite{dalalyan2012sparse, dalalyan2017further} paved the way for a characterization of rates in $\{\KL, \eu R_q\}$~\cite{cheng2018convergence, mou2022improved, durmus2019analysis, vempala2019rapid, ma2021there, chewi2021analysis, chen2022improved,zhang2023improved,altschuler2023faster}. This research program led to many fruitful insights into algorithmic properties and mathematical techniques, and inspires us to attempt the same in the constrained setting. See \S\ref{sec:related} for detail.

In purely practical terms, the $q$-R\'enyi divergence, particularly $\eu R_\infty$, has appeared in the literature most often in the context of differential privacy (DP)~\cite{mironov2017renyi}. $\eu R_q$ can be used to establish $(\epsilon, \delta)$-DP algorithms~\cite{hardt2010geometry}, while $\eu R_\infty$ can directly establish $\epsilon$-DP algorithms. This connection has been well-explored in other sampling applications~\cite{ganesh2020faster, gopi2022private, gopi2023private}. However, convergence in $\eu R_\infty$ has remained largely out of reach, and the existing algorithms are both suboptimal in complexity and limited in scope~\cite{hardt2010geometry, mangoubi2022sampling}.

\subsection{Contributions}
We close a gap in our current understanding of uniform sampling via a principled approach.
We first leverage recent results~\cite{kook2024inout} for sampling on convex bodies, proposing a generalization of~\cite{kook2024inout} that we call the $\ps$. Using uniform ergodicity results from the study of Markov chains~\cite{del2003contraction}, we show that it is possible to achieve convergence in the \emph{very strong} notion of $\eu R_\infty$ for both uniform and truncated Gaussian distributions on $\mc K$ from an $\O(1)$ warm-start without suffering any loss in complexity. By using the $\ps$ inside an annealing scheme we call $\sps$, we show that it is possible to sample uniformly from a well-rounded convex body (i.e., $\E_{X \sim \pi}[\norm{X}^2] \leq C^2 d$) with $\Otilde(C^2 d^3\polylog  \frac{1}{\varepsilon})$ queries in expectation. This convergence takes place in $\eu R_\infty$, much stronger than the prior best results in $\tv$. 

This framework is principled, and the essential benefits of the annealing scheme and sampler are made clearly evident. 
Furthermore, the approach is simple, with no need to resort to algorithmic modifications as seen in other work~\cite{hardt2010geometry, pmlr-v65-brosse17a, mangoubi2022sampling}. We present our improvements in further detail below.

\paragraph{Result 1: $\eu R_\infty$ convergence from a warm-start}

We show that, given an $\O(1)$ warm-start, there exists an algorithm which samples approximately with $\varepsilon$-accuracy in $\eu R_{\infty}$ for any $\varepsilon > 0$. The target distribution is either the uniform distribution on $\mc K$ or a truncated Gaussian of the form $\mc N(0, \sigma^2 I_d)|_{\K}$. The oracle complexity is given by the following theorem.
\begin{thm}[{Complexity of $\ps$ from a warm-start, informal version of Theorems~\ref{thm:ps-unif-Rinfty} and~\ref{thm:prox-sampler-truncated-gsn-final}}]
    Consider a target of either of the forms $\pi \propto \ind_{\mc K}$, $\pi_{\sigma^2} = \mc N(0, \sigma^2 I_d)|_{\mc K}$, where $B_1(0) \subseteq \mc K \subseteq B_D(0)$. Then, given $\pi_0$ with $\eu R_\infty(\pi_0 \mmid \pi) = \O(1)$, the $\ps$ algorithm with probability $1-\eta$ succeeds in sampling from $\nu$, $\eu R_\infty(\nu \mmid \pi) \leq \varepsilon$, requiring no more than
\[
    \O\Bigl(d^2 \msf{D}^2 \log^{\O(1)} \frac{D}{\eta \epsilon}\Bigr)\,, \qquad \msf{D}^2 := \begin{cases}
    \frac{1}{d} \vee \sigma^2 & \text{if the target is $\pi_{\sigma^2}$}\,,\\
    D^2 & \text{if the target is $\pi$}\,,
\end{cases}
\]
membership queries in expectation.
\end{thm}

The algorithm is a generalization of the $\ino$ sampler~\cite{kook2024inout} (Algorithm~\ref{alg:prox-uniform}). We call this $\ps$ (Algorithm~\ref{alg:prox-uniform}, \ref{alg:prox-gaussian}), in analogy with algorithms in unconstrained sampling~\cite{lee2021structured}.  The chief benefit of this scheme is its analytic simplicity, which allows us to extend known guarantees with little effort. More specifically, it is naturally connected to isoperimetric properties such as the log-Sobolev constant of the target. This fact will be key in establishing convergence in $\eu R_\infty$. 

\begin{algorithm}
\hspace*{\algorithmicindent} \textbf{Input:} initial point
$x_0 \sim \mu_{0}$, convex body $\mc K\subset\R^{d}$, iterations $T$, threshold $N$, and $h>0$.

\hspace*{\algorithmicindent} \textbf{Output:} $x_{T+1}$.

\begin{algorithmic}[1] \caption{$\msf{Prox\text{-}uniform}$} 
\label{alg:prox-uniform}
\FOR{$i=0,\dotsc,T$}

\STATE Sample $y_{i+1}\sim
\mc N(x_{i},hI_{d})$.\label{line:forward}

\STATE 
Repeat: Sample $x_{i+1}\sim\mc N(y_{i+1},hI_{d})$ until $x_{i+1}\in\mc K$ or \#attempts$_{i}$ $\geq N$ (declare \textbf{Failure}).\label{line:backward} 
\ENDFOR
\end{algorithmic}
\end{algorithm}

\begin{algorithm}
\hspace*{\algorithmicindent} \textbf{Input:} initial point
$x_0 \sim \mu_{0}$, convex body $\mc K\subset\R^{d}$, iterations $T$, threshold $N$, and $h>0$.

\hspace*{\algorithmicindent} \textbf{Output:} $x_{T+1}$.

\begin{algorithmic}[1] \caption{$\msf{Prox\text{-}Gaussian}$} 
\label{alg:prox-gaussian}
\FOR{$i=0,\dotsc,T$}

\STATE Sample $y_{i+1}\sim
\mc N(x_{i},hI_{d})$.\label{line:forward-gsn}

\STATE 
Repeat: Sample $x_{i+1}\sim\mc N\bpar{\frac{1}{1+h\sigma^{-2}}y,\frac{h}{1+h\sigma^{-2}}I_{d}}\cdot\ind_{\mc K}$ until $x_{i+1}\in\mc K$ or \#attempts$_{i}$ $\geq N$ (declare \textbf{Failure}).\label{line:backward-gsn} 
\ENDFOR
\end{algorithmic}
\end{algorithm}

The surprising fact is that this convergence takes place essentially in the strongest possible sense of $\eu R_{\infty}$ without any overhead. To show this, we also demonstrate a strong mixing property of $\ps$ given any point mass within $\mc K$ as the initial distribution. While we are not the only work to obtain rates in this divergence~\cite{hardt2010geometry,mangoubi2022sampling}, we are the first to show $\eu R_\infty$ guarantees with efficient rates and without resorting to algorithmic modifications.

Furthermore, the results for the truncated Gaussian are the first guarantees in $\eu R_q$ and $\eu R_\infty$ for the non-uniform, constrained case with $d^2 \sigma^2$ query complexity. While nowhere close to a comprehensive guarantee for sampling from constrained distributions $\tilde \pi \propto e^{-f}|_{\mc K}$, our results hint that this generic problem can be tackled using the same analytic framework.

\paragraph{Result 2: Warm-start generation and $\Otilde(d^3)$ uniform sampler in $\eu R_\infty$}
As a major application of the previous result, we show that, for any well-rounded convex body $\mc K$, there is an annealing algorithm with $\Otilde(d^3\polylog\frac{1}{\veps})$ complexity that achieves $\varepsilon$-accuracy in $\eu R_\infty$. Compared to $\gc$, this algorithm uses $\ps$ as a subroutine for sampling at each stage. As a result, we call it $\sps$ (Algorithm~\ref{alg:sps}). This allows us to strengthen the previous guarantees from~\cite{cousins2018gaussian} to $\eu R_\infty$-divergence, which encompasses all other divergences.

\begin{thm}[{$\eu R_{\infty}$ guarantees for uniform sampling, informal version of Theorem~\ref{thm:main-result}}]\label{thm:main-result-informal}
    Let $\pi \propto \ind_{\mc K}$ for $\mc K$ well-rounded, i.e. $\E_{X \sim \pi} [\norm{X}^2] \leq C^2 d$ while $B_1(0) \subseteq \mc K$. Then, $\sps$
    succeeds in sampling from $\nu$ with probability $1-\eta$, where $\eu R_\infty(\nu \mmid \pi) \leq \varepsilon$, requiring no more than $\O\bigl(C^2 d^3 \log^{\O(1)}\frac{1}{\eta \varepsilon}\bigr)$
    expected number of membership queries.
\end{thm}
We emphasize that our rate of $\Otilde(d^3)$ matches that of~\cite{cousins2018gaussian} while being in a much stronger metric ($\eu R_\infty$). To compare, the expected rate of the proximal sampler for unconstrained distributions $\nu$ is $\Otilde(\beta C_{\msf{PI}}(\pi) d^{3/2})$ from a feasible start~\cite[Remark 5.5]{altschuler2023faster}, where $\beta$ is the Lipschitz constant of $\nabla \log \pi$. By contrast, Theorem~\ref{thm:main-result-informal} states that, in exchange for letting $\beta = \infty$, we need to pay only an extra multiplicative factor of $d^{3/2}C_{\msf{PI}}^{-1}$ when sampling from a well-rounded convex body.

\subsection{Techniques and challenges}\label{scn:techniques}
\subsubsection{Constrained samplers}
\paragraph{Uniform sampling from a warm-start}
Before worrying about the problem of obtaining a warm start, it is necessary to establish the best possible oracle complexity for the sampling routine in $\eu R_\infty$ when given an $\O(1)$ warm-start.

Previous works such as~\cite{lovasz2006simulated,cousins2018gaussian} relied on $\har, \bw$ as underlying samplers with their annealing schemes. These possess several disadvantages. Firstly, the method of analysis is rather complicated. As for $\bw$, the best analysis (in terms of complexity) makes statements about $\sw$, which is a biased variant of $\bw$ \cite{kannan1997random,kannan2006blocking}. To implement $\bw$ with $\sw$, an additional correction step for debiasing the stationary measure is required.  This degrades rate estimates for $\sw$ to $\tv$, making the rate estimates for $\bw$ only in $\tv$ as a result. $\har$, on the other hand, is only known to mix in $\{\tv, \chi^2\}$ for uniform and exponential distributions~\cite{lovasz2006hit} and in $\tv$ for general log-concave distributions~\cite{lovasz2006fast}.

By contrast, the recent work of~\cite{kook2024inout} demonstrated that $\ino$ (Algorithm~\ref{alg:prox-uniform}) could achieve $\varepsilon$-mixing in any $\eu R_q$, requiring $\Otilde(qM d^2 \norm{\Sigma}_{\op} \polylog \textfrac{1}{\veps})$ queries in expectation from an $M$-warm start.

\paragraph{$\ino$ and the $\ps$}
The algorithm of $\ino$ (or $\msf{Prox\text{-}uniform}$) is remarkably simple. It alternately samples from two distributions, first getting $y$ from a Gaussian centred at the current point, then $z$ from a Gaussian centred at $y$ but truncated to $\mc K$ (Lines~\ref{line:forward} and~\ref{line:backward} in Algorithm~\ref{alg:prox-uniform}). This turns out to correspond to simulating the heat equation applied to the current measure, and then simulating a notion of time reversal of the heat equation (Lemma~\ref{lem:extend}). These dynamics cause any $\eu R_q$ divergence to decay exponentially, with the speed of decay dependent on the isoperimetry of the target distribution.

The only remaining challenge is to bound the wasted queries in Line~\ref{line:backward}. This follows a local conductance argument, where one quantifies the probability of $y$ landing in a set that has a hard time returning to $\mc K$. This step requires both that $\mu_0$ be warm (for reasons to be explained later) and that the variance $h =\Otilde(d^{-2})$ is appropriately small.

\paragraph{Gaussian sampling with $\ps$}
Since we will end up annealing with Gaussians, we need to demonstrate that the guarantees for $\ino$ are compatible with (truncated) Gaussian distributions $\mc N(0, \sigma^2 I_d)|_{\mc K}$. This is not necessarily trivial. While the rate estimates for $\bw$ can be adapted for Gaussians \cite{cousins2018gaussian}, those for $\har$ cannot be~\cite{lovasz2006fast}. To use $\har$ in an annealing scheme, \cite{lovasz2006simulated} resorts to exponential distributions, which have worse complexity.

In this work, we generalize the $\ino$ algorithm for truncated Gaussians on $\mc K$. This can be done by adapting the forward and backward heat flow interpretations to incorporate a potential function in the target, in analogy with works in unconstrained sampling~\cite{chen2022improved}. By the same analogy, we call our algorithm the (constrained) $\ps$.

While the calculations for the Gaussian case are somewhat involved, the approach to bounding the query complexity does not require any algorithmic modifications instead. The number of iterations needed to mix in $\eu R_q$ for the $\ps$, $\Otilde(q(d \vee d^2 \sigma^2))$, already follows as a consequence of a generic lemma shown in~\cite{kook2024inout} (restated in Lemma~\ref{lem:contraction-prox-gauss}). As for the query complexity, a careful bound on the local conductance following the same approach as~\cite{kook2024inout} (Lemma~\ref{lem:num-membership-queries}) shows that $\Otilde(M)$ membership queries suffices in each iteration. This makes $\Otilde(qM(d \vee d^2 \sigma^2))$ queries in total. The complexity matches that of~\cite{cousins2018gaussian} for Gaussians in $\tv$, and the ease of derivation underscores another strength of the $\ps$/$\ino$ framework.

We note that this is the first time that $d^3$ guarantees have been shown in $\eu R_q$ (and, as we shall see, in $\eu R_\infty$ as well) for truncated Gaussians, even if just for the specific subclass of those that are centered and isotropic. The techniques here suggest that, with more effort, the analysis could be conducted more generally for other distributions of the form $\pi = e^{-f}|_{\K}$.

\subsubsection{Warm starts and R\'enyi divergence}
\paragraph{Importance of warm starts}
We briefly digress to explain why warm starts are necessary in each of these contexts. The above samplers repeat a subroutine wherein a point is proposed, either from uniform distribution in a ball for $\bw$, and a Gaussian for $\ino$. However, na\"ively using each of these points will bias the sampling algorithm. Thus, the subroutines must contain rejection steps which correct the distribution of the proposal. To preserve our complexity bounds, we must ensure that the number of rejections is small. Otherwise, the algorithm will waste the majority of its queries without moving. 

How to control the number of rejections? If $\ino$ is initialized at the target $\pi \propto \ind_{\mc K}$ with suitable algorithmic parameters, then one can bound the expected number of rejections by $\Otilde(1)$. 
Instead, with $\exp (\eu R_\infty) \leq M$ at initialization, a straightforward calculation bounds the number of rejections by that experienced by $\pi$, multiplied by a factor of $M$. This explicitly reveals the importance of starting warm when bounding the number of rejections.

Given this context, we denote by $\msf{M} \to \msf{N}$ when a sampler starting with warmness in a metric $\msf{M}$ has convergence guarantees in a possibly weaker metric $\msf{N}$.
The guarantees of the previous samplers can be stated for the truncated Gaussian as $\bw, \har: \eu R_\infty \to \tv$ and $\ps: \eu R_\infty \to \eu R_q$. The previous discussion implies that we would require $\eu R_\infty \to \eu R_\infty$ to relay $\eu R_\infty$-guarantees across the annealing routine. Of the aforementioned works, only $\ps$ is close to achieving this. Nonetheless, closing this gap from $\eu R_q$ to $\eu R_\infty$ is far from trivial and requires the introduction of new analytical techniques, as we subsequently explain.

\paragraph{The difficulties of $\eu R_\infty$}
One would expect that convergence in $\eu R_\infty$ is extremely difficult to show. By comparison, even for a continuous-time process such as the Langevin dynamics, convergence in $\eu R_\infty$ is not known outside a few specific cases. Previous sampling guarantees in $\eu R_q$ usually involve a complexity at least linear in $q$~\cite{chewi2021analysis, chen2022improved,zhang2023improved,altschuler2023faster,kook2024inout}, which render them useless for $\eu R_\infty$.

The techniques for analyzing $q$-R\'enyi divergence usually involve constructing a Markov semigroup whose interpolation in discrete time yields the sampler of interest. The time evolution equations for this semigroup can then be used to derive exponential convergence of $\eu R_q$, with a constant depending on the isoperimetry of the stationary measure for the semigroup. This approach fails miserably for $\eu R_\infty$. In its standard representation, $\eu R_\infty$, being the supremum of the density ratio, cannot be written as the expectation of a continuous quantity, and its decay properties are difficult to establish. Instead, one must look for proof strategies that use more information about the semigroup.

\paragraph{Any-start implies $\eu R_\infty$}
It is well known in the Markov chain literature that, if a Markov chain mixes rapidly in $\tv$ from any deterministic starting point (a property known as \emph{uniform ergodicity}), then the Markov chain causes the density to contract towards that of $\pi$ in any arbitrarily strong norm. In particular, choosing the $L^\infty$-norm, it is possible to translate statements in a weak norm like ``the algorithm converges to within $\varepsilon$ in $\tv$ distance of $\pi$ within $K$ iterations, given \emph{any starting point} in the support of $\pi$'' to an unconditional statement in an \emph{extremely strong} norm --- ``the algorithm mixes in $L^\infty$ distance within $K$ iterations'' (Theorem~\ref{thm:boosting}). We call this analytical technique \textbf{boosting} for its remarkable strengthening of the performance metric.

This result had thus far been difficult to apply in most sampling settings. Consider for instance sampling from an unconstrained standard Gaussian, via some MCMC method like the Langevin diffusion. The $\tv$ can always be bounded by the $\KL$-divergence, which contracts exponentially. However, because $\pi$ is supported on $\R^n$, it is always possible to pick a starting point arbitrarily far away from the Gaussian's mode. This causes $\KL$ at initialization to be arbitrarily large. The number of iterations that are needed to converge from a point-mass at $x$ is roughly $\gtrsim \log \norm{x}$, which is not bounded for all starting points $x \in \R^d$. Thus, no finite number of iterations of the Markov chain will be sufficient to show uniform ergodicity across the entirety of $\R^d$, even though the dynamics converge exponentially quickly in any $\eu R_q$.

In passing, we also mention another notion of boosting from $\tv$ to $\eu R_\infty$. This is done via a post-processing step that can be found in the work of~\cite{mangoubi2022sampling}. Examining the guarantees of this approach, however, shows that the query complexity of the original algorithm unavoidably increases by a factor of $\poly(d)$. Thus, even if one could obtain warm-start at every step, this approach would give at minimum $\Otilde(d^5)$ query complexity\footnote{One can use the $\gc$ algorithm followed by the post-processing of \cite{mangoubi2022sampling}. The claimed complexity follows from the $\gc$'s complexity of $\Otilde(d^3\log^2\frac{1}{\epsilon})$ along with $\epsilon \lesssim \exp(-d)\veps$. See \S\ref{sec:related} for detail.} for the entire algorithm. Our approach, in contrast, does not require any algorithmic modifications, but rather extracts a convergence guarantee that is already latent in the algorithm.

\paragraph{Proximal sampler mixes from any start}
Thus, it suffices to establish convergence guarantees given any deterministic starting point within $\mc K$. When attempting to do this, a major issue arises in the analysis of the classic $\gc$, which uses $\bw$. The mixing of $\bw$ is established using a conductance-oriented proof. If the $\bw$ is analyzed directly via $s$-conductance, however, then the warmness of the initial distribution shows up directly in the number of $\bw$ steps (rather than just the query complexity) \cite{lovasz2007geometry}. This is problematic, since if we started at a bad initial point (for example, the tip of a cone), the resulting number of $\bw$ steps can be exponentially large in dimension $d$, making it impossible to invoke uniform ergodicity without degrading the rate. The other way of implementing $\bw$ through $\sw$~\cite{kannan1997random,kannan2006blocking}, requires a non-Markovian correction step, which rules out the boosting technique. Another candidate, $\har$, uses a weaker isoperimetric property (more precisely, a Cheeger inequality), so it depends poly-logarithmically on the warmness~\cite{lovasz2006hit,lovasz2006fast}. Hence, this disqualifies $\har$ as a sampler due to additional overhead of $\text{poly}(d)$ from any feasible start.

For these discrete-time samplers, one could try to directly prove a \emph{modified logarithmic-Sobolev inequality} (MLSI) of their corresponding Markov kernel (as opposed to for the stationary distribution $\pi$). While this approach can potentially bound the complexity overhead of these samplers from any feasible start by $\polylog(d)$ (or worse), they require an involved analysis of the kernel as well as the target $\pi$. In contrast, $\ps$ can be interpreted in terms of continuous forward and backward heat flows, so its mixing can be characterized by functional inequalities of just the target $\pi$. Thus, $\ps$ enjoys substantially simpler analysis, with richer connections to established functional inequalities for log-concave measures.

Using $\ps$ (Algorithm~\ref{alg:prox-uniform} and~\ref{alg:prox-gaussian}), we can indeed achieve any-start results without overhead, thereby obtaining $\eu R_\infty$ guarantees via the boosting.  
Leveraging the log-Sobolev inequality for the target, the number of iterations needs only depend doubly-logarithmically on the warmness. From a point-mass, it can be established that the warmness (after one iteration of $\ps$) is exponential in dimension, so this allows one to pay only a logarithmic factor when establishing mixing from any point-mass within $\mc K$.  These factors combined allow $\ps$ to do what $\bw, \har$ could not: obtain $\eu R_\infty$ guarantees without any computational overhead or post-processing. To our knowledge, our paper is the first work in MCMC to accomplish this, and we hope that our methods would serve as a blueprint for obtaining similar guarantees, at least in the constrained sampling setting.

\subsubsection{Proximal Gaussian cooling}

At its core, our approach in the annealing part does not deviate from~\cite{cousins2018gaussian}. The core idea is that, for a family of distributions $\{\mu_i\}_{i \in \{0, \ldots, k_*\}}$, one must balance two factors: (i) Each distribution must be $\O(1)$-warm with respect to the succeeding distribution, and (ii) the total number of distributions should not be larger than $\Otilde(d)$. The choice of Gaussian distributions gives more flexibility in this respect, since one is able to update the covariance with accelerating speed. The $\sps$ algorithm, including the choices of covariance, is given in Algorithm~\ref{alg:sps}. 
We note that $\gc$ with $\bw$ sidestepped the issue of warmness by coupling together approximate samplers. This strategy, however, reduces all the bounds down to $\tv$, and cannot be applied when examining a stronger metric.

With our approach, by combining all the aforementioned technical ingredients, we obtain a $\Otilde(d^3)$ algorithm in $\eu R_\infty$ for a well-rounded convex body. One additional benefit of our scheme is its high-level simplicity: we only need to implement the $\ps$ for different targets, which can be viewed as instances of a unified algorithm. Using this, we then follow the Gaussian-based annealing strategy.

\begin{algorithm}[H]
\hspace*{\algorithmicindent} \textbf{Input:} convex body $\mc K\subset\R^{d}$ with $0\in \K$ and $\E_{\K}[\norm{X}^2] \leq C^2d$.

\hspace*{\algorithmicindent} \textbf{Output:} approximately uniform sample $z$.

\begin{algorithmic}[1] \caption{$\sps$} \label{alg:sps}
\STATE Let $\bar{\K} = \K \cap B_{L\sqrt{d}}(0)$ with $L=C\log\frac{1}{\veps}$, and $\sigma^2:= 1/d$. Denote $\pi_{\sigma^2}:= \mc N(0,\sigma^2I_d)|_{\bar{\K}}$.
\STATE Sample $z \sim \textsf{Unif}\,(B_1(0))$.

\STATE Get $z \sim \psgauss$ with initial point $z$, target dist. $\pi_{\sigma^2}$, target accuracy $\log 2$ (in $\eu R_\infty$), and success probability $\eta/\Otilde(L^2d)$.

\WHILE{$1/d\leq \sigma^2\leq d$}
\STATE Get $z$ from $\psgauss$ with the same setup. Then update 
\[
\sigma^2 \gets \begin{cases}
    \sigma^2 \bpar{1+\frac{1}{d}} & \text{if } \frac{1}d\leq \sigma^2 \leq 1\,.\\
    \sigma^2 \bpar{1+\frac{\sigma^2}{L^2d}}& \text{if } 1\leq \sigma^2 \leq L^2d\,.
\end{cases}
\]
\ENDWHILE
\STATE Return $z\sim \psunif$ with initial point $z$, target dist. $\pi^{\bar{\K}}$, target accuracy $\varepsilon$, and success probability $\eta/\Otilde(L^2d)$
\end{algorithmic}
\end{algorithm}

\subsection{Related work}\label{sec:related}

\paragraph{Constrained samplers}

The previous approaches to constrained sampling via random walks date to~\cite{dyer1991random}. The most well-studied algorithms are $\bw$~\cite{lovasz1993random, kannan1997random} and $\har$~\cite{smith1984efficient, lovasz1999hit}. $\bw$ is a simple scheme which samples $y$ from a ball around the current point, accepting $y$ if $y \in \mc K$ and otherwise just remaining at the current point. Although it is possible to analyze $\bw$ directly~\cite{lovasz2007geometry}, usually $\bw$ refers to the the algorithm given by composing $\sw$ with rejection sampling. The best known guarantee in~\cite{kannan2006blocking} gives a complexity of $\Otilde(Md^2 D^2 \log \frac{1}{\varepsilon})$ with $M = \exp \eu R_\infty(\mu_0 \mmid \pi)$. 

Another standard algorithm is $\har$, which samples uniformly at random from a chord $\K\cap \ell$, where $\ell$ is a random line passing through the current point.
Its query complexity for uniform and exponential distributions is $\Otilde(d^2 \tr\Sigma \log \frac{M}{\varepsilon})$, where $\Sigma = \E_\pi[(X- \E X)^{\otimes 2}]$. See \S\ref{scn:techniques} for the benefits and drawbacks of these two approaches.

Both $\bw$ and $\har$ can assume that $\mc K$ is well-rounded, in the sense that $\E_\pi[\norm{X}^2] \leq C^2 d$ for a known constant $C$. A general convex body $\mc K$ can be converted to a well-rounded body via an affine transformation. Finding this transformation requires an additional algorithmic ingredient called \emph{rounding}. It has been claimed in the literature~\cite{jia2021reducing} that a randomized rounding algorithm exists with $\Otilde(d^3)$ complexity. Since this adds no computational overhead, we make the same assumption in the present work.

Apart from these general-purpose random walks, there also exist several samplers which exploit geometric information about $\mc K$, or which use stronger oracle models. The $\dw$ algorithm makes use of a self-concordant barrier function $\phi$ to draw samples from an anisotropic Gaussian (instead of isotropic as in $\bw, \ps$), and converges in $\Otilde(md)$ iterations for a convex body specified by $m$ linear constraints \cite{kannan2012random, kook2024gaussian}. Apart from this, Riemannian algorithms such as $\rhmc$ or $\rlmc$ equip $\mc K$ with a Riemannian metric, and then run a random walk using this geometry~\cite{girolami2011riemann, lee2018convergence, li2020riemannian, kook2022condition, cheng2022theory, gatmiry2023sampling}. 
Likewise, $\ml$~\cite{zhang2020wasserstein, jiang2021mirror, ahn2021efficient, li2022mirror, srinivasan2023fast} alters the geometry so that one can apply methods from unconstrained sampling, and can obtain $\Otilde(d)$ complexity with some additional assumptions. 

Constrained samplers like~\cite{bubeck2018sampling, lehec2023langevin} borrow algorithms from the unconstrained setting, and interweave them with projection steps onto $\mc K$. The best known overall complexity is $\Otilde(\frac{d^2 D^3}{\varepsilon^4})$ in terms of projection oracle queries. 
Finally, \cite{pmlr-v65-brosse17a, gurbuzbalaban2022penalized} use ``soft'' penalties rather than projections. All of the aforementioned techniques, however, require either additional assumptions on $\mc K$ or the oracle model, or are otherwise inefficient terms of complexity.

\paragraph{Annealing strategies}
Annealing as a computational tool dates back to \cite{kirkpatrick1983optimization}. Its original application was to combinatorial optimization, where the landscape of solutions is highly non-convex, and it has served as an efficient strategy for diverse problems~\cite{henderson2003theory, kalai2006simulated,delahaye2019simulated}.

The history of annealing for volume computation / uniform sampling dates back to the inception of constrained sampling, from the earliest works of~\cite{dyer1991random} through a long line of further improvements~\cite{lovasz1991compute, lovasz1990mixing, applegate1991sampling, lovasz1993random, kannan1995isoperimetric, kannan1997random}. The original incarnation of this algorithm, which remained unchanged through the listed references, samples uniformly from a sequence of bodies $\{\mc K_k\}_{1 \leq k \leq k_*}$, with $\mc K_k = \mc K \cap 2^{k/d} B_1(0)$, with $k_* = \Otilde(d)$ iterations. The distribution at each prior phase is used as a warm start for the subsequent phase. Combined with a coupling argument for the approximate sampler, the law of the sample at iteration $k_*$ can be viewed as a warm start for the uniform distribution on $\K$. While the number of phases is not severe, the bodies become increasingly difficult to sample as $k$ increases. For instance, the complexity with $\ps$ moves from $\Otilde(d^2)$ in the earliest phases to $\Otilde(d^3)$ in the latest phases, with the total complexity being $\Otilde(d^4)$ as a result.

To rectify this, \cite{lovasz2006simulated} proposed a similar multi-phase sampling strategy, but successively sampling with the exponential distributions $\mu_k \propto e^{-a_k^\T x}|_{\mc K}$ where $\norm{a_k} = 2d(1-d^{-1/2})^k$. This needs only $k_* = \Otilde(\sqrt{d})$ phases to converge. However, these exponential distributions become increasingly ill-conditioned as $k$ increases, and the warmness between $\mu_k, \mu_{k+1}$ is only in $\chi^2$. As a result, the best known complexity for obtaining a single sample with this approach is $\Otilde(d^{7/2})$. 

\cite{cousins2018gaussian} proposes~$\gc$, which avoids the pitfalls of both earlier approaches. Each stage $\mu_k$ consists of a truncated Gaussian with variance $\sigma_k^2$. As opposed to the previous annealing schemes, it is possible to accelerate the schedule (in the sense that $\sigma_k/\sigma_{k-1}$ is increasing in $k$) while still maintaining $\eu R_\infty$-warmness between $\mu_k, \mu_{k+1}$. As a result, while the total number of phases is $\Otilde(d)$, the overall complexity remains at $\Otilde(d^3)$.

\paragraph{Warm starts and R\'enyi divergence}
The convergence of samplers in R\'enyi divergences has been well studied in unconstrained sampling, beginning with results in continuous time~\cite{cao2019exponential, vempala2019rapid}. These were followed by R\'enyi guarantees of discrete time algorithms under log-Sobolev inequalities~\cite{ganesh2020faster, erdogdu2022convergence, chewi2021analysis}, Poincar\'e inequalities~\cite{lehec2023langevin, chewi2021analysis} and beyond~\cite{mousavi2023towards}. These algorithms are vital for warm starts in the unconstrained field as well, particularly for Metropolis adjusted algorithms~\cite{dwivedi2018log, chen2018fast, chewi2021optimal,chen2023does}. In fact, the study of $\eu R_q$ divergences and their connection to warm starts directly led to the present state-of-the-art algorithms in unconstrained sampling~\cite{altschuler2023faster}.

Within constrained sampling, 
the work of~\cite{hardt2010geometry} proposed one method for obtaining guarantees in $\eu R_\infty$ on transformations of the unit ball. This approach is based on grid walk, and unrelated to that in the current work. It has been extended by~\cite{bassily2014private} to more general distributions. Secondly,~\cite{mangoubi2022sampling} show a boosting scheme which mollifies the convex body by convolving it with a ball. Such a technique uses a sampler with $\O(e^{-d} \varepsilon)$-guarantees in $\tv$ as an inner loop of the routine, which adds at least $\poly(d)$ overhead when applied to a sampler with $\polylog \frac{1}{\varepsilon}$ dependence. 

The technique of proving $L^\infty$ bounds using uniform ergodicity was well known in the study of Markov chains. Its history dates perhaps back to the works of Markov himself~\cite{markov1906extension}, among other venerable works~\cite{yosida1941operator,doob1953stochastic, meyn2012markov}. It is connected with Doeblin's minorization condition and other fundamental properties of Markov chains~\cite[\S16]{meyn2012markov}, and has perhaps been most succinctly stated in~\cite[Proposition 3.23]{Rudolf2011ExplicitEB}. The uniform ergodicity property, however, is difficult to establish, barring some exceptions such as exact Gibbs sampling or samplers on the lattice~\cite{lindsten2015uniform, wang2017geometric}. As far as we know, we are the first to apply it in the constrained sampling setting. 

\subsection{Organization}
The remainder of this paper is organized as follows. In \S\ref{sec:prelim}, we introduce key notions that will be used in the paper, and in \S\ref{scn:tv-to-renyi-boosting} elucidate the boosting technique, which converts any-start guarantees in $\tv$ to a contraction in $\eu R_\infty$. We then illustrate how it can be applied for uniform distributions on $\mc K$, as a teaser for \S\ref{scn:trunc-gsn}, where we develop guarantees for truncated Gaussians. Finally, we put these all together in \S\ref{scn:successive-prox} to obtain our final, $\Otilde(d^3)$ uniform sampling guarantees before concluding.

\section{Preliminaries}\label{sec:prelim}
\subsection{Notation}
Unless otherwise specified, $\norm{\cdot}$ denotes the $2$-norm on $\R^d$ and the operator norm on $\R^{d \times d}$. The notation $a= \mc{O}(b), a \lesssim b$ will be used to signify that $a \leq cb$ for $c > 0$ a universal constant, while $a \gtrsim b, a = \Omega(b)$ will denote $a \geq cb$. $a = \Theta(b)$ is used when $a = \Omega(b), a= \O(b)$ simultaneously. The notation $\Otilde(b)$ will mean $a = \mc{O}(b\, \polylog(b))$, and likewise for $\widetilde{\Theta}, \widetilde{\Omega}$. Finally, we conflate a measure with its density where there is no confusion. 

\subsection{Basic notions}
Before proceeding, we reiterate our computational model.
\begin{defn}[Membership oracle]\label{def:membership-oracle}
    We are given a convex body $\mc K$ which has $B_1(0) \subseteq \mc K \subseteq B_D(0) \subset \R^d$ for some $D > 0$. We assume access to a \emph{membership oracle}, which, given a point $x \in \R^d$, answers \emph{Yes} or \emph{No} to the query ``Is $x \in \mc K$?''
\end{defn}
Where not otherwise specified, we will write $\pi = \frac{1}{\vol(\mc K)} \ind_{\mc K}$ for the uniform distribution on $\mc K$. We introduce the following metrics between probability measures.
\begin{defn}[Distance and divergence] \label{def:p-dist}
    Let $\mu, \nu$ be two probability measures on $\R^d$. Their \emph{total variation} distance is given by
    \[
        \norm{\mu-\nu}_{\msf{TV}} := \sup_{B \in \mc B(\R^d)} \abs{\mu(B) - \nu(B)}\,,
    \]
    where $\mc B(\R^d)$ is the collection of all Borel measurable subsets of $\R^d$. The \emph{$q$-R\'enyi divergence} is defined as
    \[
        \eu R_q(\mu \mmid \nu) := \frac{1}{q-1} \log \int \bigl(\frac{\D \mu}{\D \nu}\bigr)^q\, \D \nu\,,
    \]
    if $\mu \ll \nu$, and $ \eu R_q(\mu \mmid \nu) := \infty$ otherwise. In the limit $q \to 1$, it converges to the $\KL$ divergence,
    \[
        \KL(\mu \mmid \nu) := \int \log \frac{\D \mu}{\D \nu} \, \D \mu\,,
    \]
    if $\mu \ll \nu$, and again $\KL(\mu \mmid \nu) := \infty$ otherwise. The $\chi^2$-divergence is defined as $\chi^2(\mu \mmid \pi) := \exp(\eu R_2(\mu \mmid \nu)) - 1$. Finally, the limit $\lim_{q \to \infty} \eu R_q(\mu \mmid \nu)$ can be written as
    \[
        \eu R_\infty(\mu \mmid \nu) := \esssup_{\nu} \log \frac{\D \mu}{\D \nu}\,.
    \]
\end{defn}
The $\eu R_\infty$ distance is especially important for us, through the concept of warmness.
\begin{defn}[Warmness]
    For $\mu \ll \nu$, we denote the warmness $M$ of $\mu$ with respect to $\nu$ as 
    \[
        M := \esssup_\nu \frac{\D \mu}{\D \nu} = \exp \bigl(\eu R_\infty(\mu \mmid \nu)\bigr)\,.
    \]
    Alternatively, if the above holds, then we also say that $\mu$ is an $M$-warm start for $\nu$.
\end{defn}

The following inequality for R\'enyi divergences will also be useful.
\begin{lem}[Data-processing inequality]\label{lem:dpi}
    Let $\mu, \nu$ be probability measures, $P$ a Markov kernel, and $q \geq 1$. Then,
    \[
        \eu R_q(\mu P \mmid \nu P) \leq \eu R_q (\mu \mmid \nu)\,.
    \]
\end{lem}

\subsection{Functional inequalities for constrained distributions}
The following functional inequalities, also known as isoperimetry inequalities, on the target distribution are vital for the analysis of sampling algorithms, being equivalent to the exponential mixing of a broad class of Markov chains.

\begin{defn}[Log-Sobolev inequality]
    A probability measure $\nu$ satisfies a \emph{log-Sobolev inequality} ($\msf{LSI}$) with parameter $C_{\msf{LSI}}(\nu)$ if for all smooth functions $f: \R^d \to \R$,
    \begin{equation}\tag{LSI}\label{eq:lsi}
        \Ent_\nu(f^2) \leq 2 C_{\msf{LSI}}(\nu)\, \E_\nu[\norm{\nabla f}^2]\,,
    \end{equation}
    with $\Ent_\nu(f^2) := \E_\nu[f^2 \log f^2] - \E_\nu[f^2] \log \E_\nu[f^2]$.
\end{defn}

A weaker form of isoperimetry that is implied by the above is the Poincar\'e inequality.
\begin{defn}[Poincar\'e inequality]
    A probability measure $\nu$ satisfies a \emph{Poincar\'e inequality} ($\msf{PI}$) with parameter $C_{\msf{PI}}(\nu)$ if for all smooth functions $f: \R^d \to \R$,
    \begin{equation}\tag{PI}\label{eq:pi}
        \Var_\nu f \leq C_{\msf{PI}}(\nu)\, \E_\nu[\norm{\nabla f}^2]\,,
    \end{equation}
    with $\Var_\nu f := \E_\nu[(f-\E_\nu f)^2]$.
\end{defn}

The following lemma summarizes the functional inequalities which are satisfied by (truncated) Gaussians and uniform distributions on $\mc K$, due to \cite{klartag2023logarithmic, bakry2014analysis}. We refer readers to \cite[Appendix C]{kook2024inout}.

\begin{lem}\label{eq:isoperimetry-dist}
    Let $\mc K\subset \Rd$ be a convex body with diameter $D$. Then, if $\pi \propto \ind_{\mc K}$ is the uniform distribution on $\mc K$,
    \[  
        C_{\msf{LSI}}(\pi) \lesssim D^2 \qquad \text{and} \qquad C_{\msf{PI}}(\pi) \lesssim \norm{\Sigma} \log d\,,
    \]
    where $\Sigma = \E_{X \sim \pi}[(X- \E_{\pi} X)(X- \E_{\pi} X)^\T]$ is the covariance of $\pi$. For a Gaussian $\pi_{\sigma^2} = \mc N(0, \sigma^2 I_d)|_{\mc K}$, we have
    \[
        C_{\msf{PI}}(\pi_{\sigma^2})\leq \clsi(\pi_{\sigma^2}) \leq \sigma^2\,.
    \]
\end{lem}

\section{Total variation to R\'enyi infinity via LSI}\label{scn:tv-to-renyi-boosting}

In this section, we first review the work of~\cite{Rudolf2011ExplicitEB}, which allows us to bound the distance between the iterates of a Markov chain and their stationary measure in $L^\infty$ (and so in $\eu R_\infty$) by the worst-case distance in $\msf{TV}$ from \textbf{any start}. Then we reveal a connection between \eqref{eq:lsi} and convergence from any start (uniform ergodicity), demonstrating a way to leverage this powerful boosting technique.
As a concrete example, we provide a $\eu R_\infty$-guarantee of $\ps$ for uniform sampling over a convex body \textbf{without incurring} additional factors to the state-of-the-art complexity.

\subsection{Strong contraction of a Markov chain}

We recall standard notion in the theory of Markov semigroups.

\begin{defn}[Markov kernel]
    Let $(\Omega, \mc F)$ be a measurable space. A \emph{Markov kernel} $P: \Omega \times \mc F \to [0, 1]$ satisfies
    \begin{enumerate}
        \item for each $x \in \Omega$, the map $P(\cdot | x):=P(x, \cdot)$ from $\mc F$ to $[0,1]$ is a probability measure on $(\Omega, \mc F)$.
        \item for each $E \in \mc F$, the map $P(\cdot, E)$ from $\Omega$ to $[0,1]$ is $\mc F$-measurable.
    \end{enumerate}
\end{defn}

Given a probability measure $\mu$ over $\Omega$ and function $f:\Omega\to\R$, a Markov kernel $P$ acts on $\mu$ and $f$ to yield objects $\mu P$ and $Pf$ defined by 
\begin{align*}
\mu P(\cdot) &:= \int_{\Omega}P(\cdot|x)\,\mu(\D x)\,,\\
Pf(x) &:= \E_{Y\sim P(\cdot|x)}[f(Y)]=\int_\Omega f(y)\, P(\D y|x)\,.
\end{align*}
A probability measure $\pi$ is called \emph{stationary} for $P$ if $\pi P = \pi$.

In Markov semigroup theory, it is of major importance to study the \emph{contractivity} of a Markov kernel, since it quantifies the convergence rate of its corresponding Markov chain toward the stationary distribution $\pi$. This contractivity can be captured via the contraction coefficient of $P$ defined by
\[
\norm P_{L^{p}\to L^{p}}:= \sup_{0\neq f\in L^p_0}\frac{\norm{Pf}_{L^p}}{\norm f_{L^p}}\,,
\]
where $\norm{f}_{L^p} := \norm{f}_{L^p(\pi)} = (\E_\pi[|f|^p])^{1/p}$ and $L^p_0:= \{f: \E_\pi[|f|^p] <\infty, \E_\pi f=0\}$.

The most classical setting that has been studied is the $L^2(\pi)$ space, whose contraction coefficient is given by $\gamma$ if the \emph{spectral gap} of the Markov kernel $P$ is $1-\gamma$.
By substituting $f=\frac{\D \mu}{\D\pi}-1$, it is also possible to quantify the convergence rate of a Markov chain in $\chi^2$-divergence with reference to the same constant $\gamma$.

This classical setting, however, is not sufficient for understanding $\eu R_\infty$, since we only have the one-sided bound $\eu R_2 = \log(1+\chi^2) \leq \eu R_\infty =\log\esssup |\frac{\D \mu}{\D \pi}| = \log \norm{\frac{\D \mu}{\D \pi}}_{L^\infty}$. Instead, it is natural to study the contraction coefficient in $L^\infty(\pi)$, defined by
\[
\norm P_{L^{\infty}\to L^{\infty}}:=\sup_{0\neq f\in L^\infty_0}\frac{\norm{Pf}_{L^\infty}}{\norm f_{L^\infty}}\,,
\]
where $\norm{f}_{L^\infty} := \norm{f}_{L^\infty(\pi)} = \esssup_{\pi}\abs{f}$ and $L^{\infty}_0:=\{f:\esssup_{\pi}\abs{f}<\infty, \E_\pi f=0\}$.
One observes that $L^\infty\to L^\infty$ contractivity implies the \emph{uniform ergodicity} of a Markov chain, and that the opposite inequality also holds due to~\cite[Proposition 3.23]{Rudolf2011ExplicitEB}.

\begin{prop}[{\cite[Proposition 3.23]{Rudolf2011ExplicitEB}}] 
Let $P$ be a Markov kernel that is reversible with respect to the stationary distribution $\pi$. Then,
\[
\norm{P^{n}-1_{\pi}}_{L^{\infty}\to L^{\infty}}\leq 2 \esssup_{x}\norm{P^{n}(\cdot|x)-\pi}_{\tv}\,,
\]
where $1_{\pi}$ is the operator defined by $1_{\pi}(f)=\E_{\pi}[f]$.
\end{prop}

We can now deduce an explicit $L^\infty$-convergence result for a reversible Markov chain as follows. 

\begin{thm}\label{thm:boosting}
Consider a Markov chain with kernel $P$, initial distribution $\mu$ and stationary distribution $\pi$. Then,
\[
\bnorm{\frac{\D\mu P^{n}}{\D\pi}-1}_{L^{\infty}}\leq\bnorm{\frac{\D\mu}{\D\pi}-1}_{L^{\infty}}\cdot2\esssup_{x}\norm{P^{n}(\cdot|x)-\pi}_{\tv}\,
\]
In particular, $\eu R_\infty(\mu P^n \mmid \pi) \leq \bnorm{\frac{\D\mu}{\D\pi}-1}_{L^{\infty}}\cdot2\esssup_{x}\norm{P^{n}(\cdot|x)-\pi}_{\tv}$.
\end{thm}

\begin{proof}
Setting $f=\frac{\D\mu}{\D\pi}-1$, we have 
\[
\norm{P^{n}-1_{\pi}}_{L^{\infty}\to L^{\infty}}\geq\frac{\norm{P^{n}f-1_{\pi}(f)}_{L^{\infty}}}{\bnorm{\frac{\D\mu}{\D\pi}-1}_{L^{\infty}}}=\frac{\norm{P^{n}f}_{L^{\infty}}}{\bnorm{\frac{\D\mu}{\D\pi}-1}_{L^{\infty}}}=\frac{\bnorm{\frac{\D\mu P^{n}}{\D\pi}-1}_{L^{\infty}}}{\bnorm{\frac{\D\mu}{\D\pi}-1}_{L^{\infty}}}\,,
\]
where the last equality follows from $P(\frac{\D\mu}{\D\pi}) = \frac{\D(\mu P)}{\D\pi}$ due to the reversibility of $P$. Therefore, 
\[
\bnorm{\frac{\D\mu P^{n}}{\D\pi}-1}_{L^{\infty}}\leq\bnorm{\frac{\D\mu}{\D\pi}-1}_{L^{\infty}}\cdot2\esssup_{x}\norm{P^{n}(\cdot|x)-\pi}_{\tv}\,.
\]
The R\'enyi-infinity bound immediately follows from $\log(1+x)\leq x$.
\end{proof}

\subsection{LSI to uniform ergodicity without overhead}

We now need to control the $\tv$-distance \textbf{uniformly} over any initial point $x \in \Omega$. That is, one should find the iteration number $n$ of the Markov chain such that $\norm{\delta_xP^n - \pi}_{\tv} \lesssim \varepsilon$ for almost every $x\in \Omega$, where $\delta_x$ denotes the Dirac measure at $x$. One expects that it is impossible in general to bound this quantity for arbitrary Markov chains and stationary measures, but we can get around this in our current setting, which only considers probability measures with compact support.

Mixing rates of many convex bodies samplers have been studied when $\pi$ satisfies a Cheeger isoperimetric inequality (which is equivalent to a Poincar\'e inequality for log-concave distributions). For comparison, standard choices of Markov kernel in the unconstrained setting (such as the Langevin or underdamped Langevin dynamics) relate Poincar\'e inequalities for $\pi$ to the convergence of the sampler in $\chi^2$, and have theoretical guarantees typically given by
\[
\chi^2(\delta_x P^n \mmid \pi) \lesssim \exp\bpar{-\frac{n}{\cpi(\pi)}}\, \chi^2(\delta_x P^1 \mmid \pi)\,.
\]
As $2\,\tv^2\leq \chi^2$, one can achieve $\veps$-$\tv$ in $\cpi(\pi) \log(\chi^2(\delta_x P^1 \mmid \pi) /\veps^2) \lesssim \cpi(\pi) \log(\norm{\frac{\D(\delta_x P^1)}{\D\pi}}_{L_\infty} /\veps)$ iterations. However, the initial $\chi^2$ is typically $\exp(\Omega(d))$, so under \eqref{eq:pi} this approach ends up incurring an additional overhead of $\Omega(d)$ to the mixing rate. 

While it would already be impressive to boost a mixing rate from the weakest metric ($\tv$) to the strongest metric ($\eu R_\infty$) under \eqref{eq:pi} at the cost of additional $\poly(d)$, one can achieve this with only additional $\polylog(d)$ factors under \eqref{eq:lsi}. Just as a Poincar\'e inequality for $\pi$ is normally sufficient to imply the exponential convergence of a corresponding Markov process to $\pi$ in $\chi^2$, the log-Sobolev inequality is equivalent to exponential convergence of many process in entropy (or in $\KL$). Hence, under \eqref{eq:lsi}, theoretical guarantees of Markov chains are typically of the form
\[
\KL(\delta_x P^n \mmid \pi) \lesssim \exp\bpar{-\frac{n}{\clsi(\pi)}}\, \KL(\delta_x P^1 \mmid \pi)\,.
\]
Since we know that $2\,\tv^2 \leq \KL$ (CKP inequality), $\veps$-$\tv$ can be achieved after $\clsi(\pi) \log(\KL(\delta_x P^1 \mmid \pi)/\veps) \lesssim  \clsi(\pi) \log(\frac{1}{\veps}\log\norm{\frac{\D(\delta_x P^1)}{\D\pi}}_{L_\infty})$ iterations. As noted earlier, this initial distance is at worst $\exp(d^{\O(1)})$, which only results in additional $\polylog(d)$ factors after evaluating the double logarithm. Therefore, one can boost from $\tv$ to $\eu R_\infty$ with only logarithmic overhead through \eqref{eq:lsi}.

\subsubsection{R\'enyi-infinity guarantees for uniform sampling under warm start}

We demonstrate the effectiveness of this boosting technique using \eqref{eq:lsi} for the uniform distribution. Thereby we achieve a $\eu R_\infty$-guarantee of uniform sampling without compromising the well-known mixing rate $\Otilde(d^2D^2\log\frac{1}{\veps})$ of $\bw$ and $\har$ (given in terms of $\tv$ and $\chi^2$ respectively).

To this end, it is useful to work with a sampler whose mixing is well understood under several functional inequalities such as \eqref{eq:pi} and \eqref{eq:lsi}. Recently, \cite{kook2024inout} studies $\ps$\footnote{In the original work, it is called the $\ino$, inspired by the geometric behavior of the sampler. We call it $\ps$ instead, since this geometric behaviour is not as clear when working with arbitrary target distributions.} for uniform distributions under those functional inequalities, via calculations following the heat flow and its $\pi$-adjoint.
In particular, \cite{kook2024inout} already establishes $\eu R_q$-guarantees (with $q<\infty$) of the $\ps$ with respect to the uniform distribution $\pi$, given any starting measure.

\begin{lem}[{\cite[Corollary 28]{kook2024inout}}]\label{lem:prox-unif-any-start}
For any $\veps\in(0,1)$ and convex body $\mc K\subset\Rd$ with diameter $D$, let $P$ be the Markov kernel of the $\ps$ with variance $h$. For given $x\in\mc K$, let $\mu_{x}^{n}:=\delta_{x}P^{n}$ be the law of the $n$-th iterate of the $\ps$, and $\pi$ be the uniform distribution over $\mc K$. Then, $\eu R_{q}(\mu_{x}^{n}\mmid\pi)\leq\veps$ for $n=\Otilde\bpar{qh^{-1}\clsi(\pi)\log\frac{d+h^{-1}D^{2}}{\veps}}$.
\end{lem}

Combining this with the boosting technique, we obtain a $\eu R_\infty$-guarantee for uniform sampling:

\begin{lem}\label{lem:prox-sampler-unif-final}
For a convex body $\mc K\subset\Rd$ with diameter $D$, assume that an initial distribution $\mu$ is $M$-warm with respect to the uniform distribution $\pi$ over $\K$. 
For any $\veps\in(0,1)$, the $\ps$ with variance $h$ achieves $\eu R_{\infty}(\mu P^{n}\mmid\pi)\leq\veps$ (or $1-\veps\leq \frac{\D\mu P^{n}}{\D\pi}\leq1+\veps$ on $\mc K$) after $n=\Otilde\bpar{h^{-1}D^{2}\log\frac{M(d+h^{-1}D^{2})}{\veps}}$ iterations.
\end{lem}

\begin{proof}
Using $2\norm{\cdot- \pi}_{\tv}^{2}\leq\KL(\cdot\mmid\pi)=\lim_{q\downarrow1}\eu R_{q}(\cdot\mmid\pi)$ and $\clsi(\pi)=\O(D^2)$, one obtains that after $n\gtrsim h^{-1}D^{2}\log\frac{M(d+h^{-1}D^{2})}{\veps}$ iterations,
\[
\sup_{x\in\mc K}\norm{\mu_{x}^{n}-\pi}_{\tv}\leq\frac{\veps}{M}\,.
\]
By Theorem~\ref{thm:boosting} with $\bnorm{\frac{\D\mu}{\D\pi}-1}_{L^{\infty}}\leq M$, we have $\bnorm{\frac{\D\mu P^{n}}{\D\pi}-1}_{L^{\infty}}\leq\veps$ and $\eu R_{\infty}(\mu P^{n}\mmid\pi)\leq\veps$.
\end{proof}

\begin{rmk}
    We note that the above guarantee bounds the \emph{iteration number} for mixing in $\eu R_\infty$, not the query complexity, through any-start guarantees in $\tv$. The query complexity needed to attain $\varepsilon$ $\tv$-distance from any start in our implementation can potentially be exponential in $d$. This, however, will not be relevant to our results since the kernel $P$ only captures the accepted proposals. This way, we can view the result above as merely extracting a latent property of the algorithm, which is not dependent on details of its implementation. For Metropolized algorithms such as $\bw$ however, the kernel $P$ will need to take the number of rejections into account, and the dependence on $d$ for any-start guarantees can potentially scale poorly as a result.
\end{rmk}

Lastly, combined with the query complexity of implementing each step of the $\ps$ for uniform distributions, we obtain a guarantee on the query complexity for achieving $\eu R_\infty$-mixing for uniform sampling.

\begin{thm}\label{thm:ps-unif-Rinfty}
    For any $\eta,\varepsilon \in (0,1)$, $n\in \mathbb{N}$ defined below, and convex body $\K$ given by a well-defined membership oracle, the $\ps$ (Algorithm~\ref{alg:prox-uniform}) with $h = (2d^2\log\frac{9nM}{\eta})^{-1}$, $N = \Otilde(\frac{nM}{\eta})$, and initial distribution $\mu_0$ $M$-warm with respect to $\pi$ the uniform distribution over $\K$ achieves $\eu R_\infty(\mu_n \mmid \pi)\leq \veps$ after $n = \Otilde(d^2 D^2 \log^2 \frac{M}{\eta \veps})$ iterations, where $\mu_n$ is the law of the $n$-th iterate.
    With probability $1-\eta$, the algorithm iterates this many times without failure, using $\Otilde(M d^2D^2\log^6\frac{1}{\eta \veps})$ expected number of membership queries in total.
\end{thm}

\begin{proof}
    By \cite[Lemma 14]{kook2024inout}, if one takes variance $h = (2d^2\log\frac{9nM}{\eta})^{-1}$ and threshold $N = \Otilde(\frac{nM}{\eta})$, then for each iteration the expected number of membership queries is $\Otilde(M\log\frac{1}{\eta\veps})$, and the failure probability is at most $\eta/n$. By Lemma~\ref{lem:prox-sampler-unif-final}, the $\ps$ should iterate $n=\Otilde(d^2D^2\log\frac{M}{\eta\veps})$ times to output a sample whose law is $\veps$-close to $\pi$ in $\eu R_\infty$. Therefore, throughout the $n$ outer iterations, the failure probability of the $\ps$ is at most $\eta$ by a union bound, and the total expected number of queries is $\Otilde(Md^2D^2\log\frac{M}{\eta\veps})$.
\end{proof}

\section{Uniform ergodicity of proximal sampler for truncated Gaussians}\label{scn:trunc-gsn}
In the $\gc$ algorithm  \cite{cousins2018gaussian}, a truncated Gaussian $\pi_{\sigma^2}:=\mc N(0, \sigma^2 I_d)|_{\K}$ for a convex body $\K$ serves as an annealing distribution, where the variance $\sigma^2$ is gradually increased across phases. In particular, $\gc$ uses a Metropolized $\bw$ to sample such a truncated Gaussian. Its convergence rate is quantified in $\tv$ through a Cheeger isoperimetry of the truncated Gaussian (namely, a Poincar\'e inequality).

In order for us to properly carry the warmness across phases, we must need a $\eu R_\infty$-guarantee for sampling truncated Gaussians. Just as we established $\eu R_\infty$ guarantees for uniform sampling via  Theorem~\ref{thm:boosting} and \eqref{eq:lsi} for the uniform distribution, it is natural to propose a sampler for a truncated Gaussian whose convergence rate can be concisely related to $\clsi(\pi_{\sigma^2})$. Thus, in \S\ref{sec:prox-gauss-warm} we analyze the $\ps$ for $\pi_{\sigma^2}$, providing its convergence rate in terms of $\clsi(\pi_{\sigma^2})$ through a heat flow perspective just as in \cite{kook2024inout}. Then in \S\ref{sec:prox-gauss-anystart} we establish the uniform ergodicity of the $\ps$ for $\pi_{\sigma^2}$, deducing the $\eu R_\infty$ guarantee in Theorem~\ref{thm:ps-unif-Rinfty}.

\subsection{Convergence analysis}\label{sec:prox-gauss-warm}

Compared with the $\ps$ for uniform distributions, the $\ps$ for truncated Gaussians (Algorithm~\ref{alg:prox-gaussian}) requires a different backward step (i.e., the implementation of RGO) while using the same forward step (i.e., Gaussian step).

One iteration of the $\ps$ for $\mc N(0,\sigma^{2}I_{d})|_{\mc K}$ consists of two steps:
\begin{itemize}
\item (Forward step) $y\sim\pi^{Y|X}(\cdot|x)=\mc N(x,hI_{d})$.
\item (Backward step) 
\[x\sim\pi^{X|Y}(\cdot|y)\propto\exp\bigl(-\frac{1}{2\sigma^{2}}\norm x^{2}-\frac{1}{2h}\norm{y-x}^{2}\bigr)\cdot\ind_{\mc K}(x)\propto\mc N\Bpar{\frac{1}{1+h\sigma^{-2}}y,\frac{h}{1+h\sigma^{-2}}I_{d}}\Big|_{\K}\,.
\]
\end{itemize}
To implement the backward step, we use rejection sampling with the proposal $\mc N\bpar{\frac{1}{1+h\sigma^{-2}}y,\frac{h}{1+h\sigma^{-2}}I_{d}}$, accepting if the proposal lies inside of $\K$.
Then the expected number of trials for the first success is
\[
\frac{1}{\ell(y)}:=\frac{\bpar{2\pi(h^{-1}+\sigma^{-2})^{-1}}^{d/2}}{\int_{\mc K}\exp\bpar{-\frac{\sigma^{-2}+h^{-1}}{2}\norm{x-h^{-1}(\sigma^{-2}+h^{-1})^{-1}y}^{2}}\,\D x}\,.
\]
We can write down the density of $\pi^{Y}$ as follows:
\begin{align}
\pi^Y(y) 
& =\frac{\int_{\mc K}\exp\bpar{-\frac{1}{2\sigma^{2}}\norm x^{2}-\frac{1}{2h}\norm{x-y}^{2}}\,\D x}{(2\pi h)^{d/2}\int_{\mc K}\exp\bpar{-\frac{1}{2\sigma^{2}}\norm x^{2}}\,\D x}\nonumber \\
& =\frac{\int_{\mc K}\exp\bpar{-\frac{1}{2}(\sigma^{-2}+h^{-1})\norm{x-h^{-1}(\sigma^{-2}+h^{-1})^{-1}y}^{2}}\,\D x}{(2\pi h)^{d/2}\int_{\mc K}\exp\bpar{-\frac{1}{2\sigma^{2}}\norm x^{2}}\,\D x}\,\exp\bpar{-\frac{1}{2(h+\sigma^{2})}\norm y^{2}}\label{eq:pi-y}\\
& =\frac{(1+h\sigma^{-2})^{-d/2}\,\ell(y)}{\int_{\mc K}\exp\bpar{-\frac{1}{2\sigma^{2}}\norm x^{2}}\,\D x}\,\exp\bpar{-\frac{1}{2(h+\sigma^{2})}\norm y^{2}}\,.\label{eq:density-ell}
\end{align}

\subsubsection{Mixing analysis}

The two steps within $\ps$ have continuous interpolation via the forward and backward heat flow, so their mixing guarantees can be naturally associated with functional inequalities (e.g.~\eqref{eq:pi} and~\eqref{eq:lsi}) for a target distribution. Such a mixing result for the $\ps$ has already been established for unconstrained distributions $\pi^X\propto\exp(-V)$~\cite{chen2022improved}, for which the $\ps$ can be generalized as follows: for $\pi(x,y)\propto \exp\bpar{-V(x)-\frac{1}{2h}\norm{x-y}^2}$, repeat (i) $y_{i+1} \sim \pi^{Y|X=x_i}=\mc N(x_i,hI_d)$ and (ii) $x_{i+1} \sim \pi^{X|Y=y_{i+1}}$.

\begin{prop}[{\cite[Theorem 3]{chen2022improved}}]\label{prop:prox-contraction}
Let $\mu_{k}^{X}$ be the law of the $k$-th output of $\ps$ with initial distribution $\mu_{0}^{X}$. Then, for any $q\geq1$, 
\[
\eu R_{q}(\mu_{k}^{X}\mmid\pi^{X})\leq\frac{\eu R_{q}(\mu_{0}^{X}\mmid\pi^{X})}{\bpar{1+h/\clsi(\pi^X)}^{2k/q}}\,.
\]
\end{prop}
This has been further extended to constrained distributions, including the truncated Gaussian, under only mild additional assumptions. 
\begin{lem}[{\cite[Lemma 22]{kook2024inout}}]\label{lem:extend}
Let $\nu$ be a measure, absolutely continuous with respect to the uniform distribution over $\K$, and $\mu_0$ an arbitrary measure. The forward and backward heat flow solutions given by
\begin{align*}
    \partial_t \mu_t &= \frac{1}{2} \Delta \mu_t \,,\\
    \partial_t \mu_{t}^{\leftarrow} &= -\Div\bigl(\mu_t^\leftarrow \nabla \log (\nu P_{h-t})\bigr)+ \frac{1}{2} \Delta \mu_t^\leftarrow\quad\text{with }\mu_0^\leftarrow = \mu_h\,,
\end{align*}
admit solutions on $(0,h]$, and the weak limit $\lim_{t \to h} \mu_t^\leftarrow = \mu_h^\leftarrow$ exists for any initial measure $\mu_0$ with bounded support. Moreover, for any $q$-R\'enyi divergence with $q \in (1, \infty)$,
\[
    \eu R_q(\mu_{h}^\leftarrow \mmid \nu) \leq \lim_{t \downarrow 0} \eu R_q(\mu_{h-t}^\leftarrow \mmid \nu_t)\,.
\]
\end{lem}

We will set $\nu=\mc N(0,\sigma^2 I_d)|_{\K}$ in above. It turns out that the solutions to the two equations above give precisely the laws of $\mu_k^Y, \mu_{k+1}^X$ when starting at $\mu_k^X$. Secondly, it is well-known in \cite{bakry2014analysis} that a truncated Gaussian $\nu$ has $\clsi(\nu) \leq \sigma^{2}$, as truncation to a convex set only improves the log-Sobolev constant.
Then we can derive a contraction result of the $\ps$ for $\nu$ in $\eu R_q$ for any $q \geq 1$, emulating the proof for uniform distributions in \cite[Lemma 23]{kook2024inout}. First, convolve the truncated Gaussian with Gaussian with very small variance $\epsilon$, which generates a smooth and unconstrained distribution. Then invoke the contraction result in Proposition~\ref{prop:prox-contraction}, and use lower semicontinuity of the $\eu R_q$ divergence to conclude when sending $\varepsilon \to 0$.

\begin{lem}\label{lem:contraction-prox-gauss}
    Let $\mu^X_k$ be the law of the $k$-th output of the $\ps$ with initial distribution $\mu_0^X$ and target $\pi^X = \mc N(0,\sigma^2 I_d)|_{\K}$.
    Then, for any $q \geq 1$,
    \[
    \eu R_q(\mu_k^X\mmid \pi^X) \leq \frac{\eu R_q(\mu_0^X \mmid \pi^X)}{(1+h\sigma^{-2})^{2k/q}}\,.
    \]
\end{lem}

\begin{proof}
For small $\epsilon>0$, as $\mu_\epsilon = (\mu_0^X)_\epsilon = \mu_0^X * \mc N(0,\epsilon I_d)$ and $\pi_\epsilon = \pi^X * \mc N(0, \epsilon I_d)$ are $C^\infty$-smooth, we can invoke the decay result with step size $h-\epsilon$ in Proposition~\ref{prop:prox-contraction}.
Thus, for contraction constants $C_{\epsilon}=(1+\frac{h-\epsilon}{\sigma^{2}+\epsilon})^{-2/q}$ (since $\clsi\bigl(\nu * \mc N(0, \epsilon I_d)\bigr)\leq \clsi(\nu)+\epsilon$ in general), it follows that
\[
\eu R_q(\mu_{h-\epsilon}^\leftarrow \mmid \pi_\epsilon) \leq C_{\epsilon}\cdot \eu R_q(\mu_{\epsilon} \mmid \pi_\epsilon)\leq C_{\epsilon} \cdot \eu R_q(\mu_0 \mmid \pi_0)\,,
\]
where we used the data-processing inequality (Lemma~\ref{lem:dpi}) for the last inequality. By the lower semicontinuity of $\eu R_q$ as noted earlier, sending $\epsilon\to 0$ leads to
\[
\eu R_q(\mu_1^X \mmid \pi^X) = \eu R_q(\mu_{h}^\leftarrow \mmid \pi_0) \leq C\cdot \eu R_q(\mu_0 \mmid \pi_0) = C\cdot \eu R_q(\mu_0^X \mmid \pi^X)\,.
\]
Repeating this argument $k$ times completes the proof.
\end{proof}

\subsubsection{Per-step guarantees}

We can find an \emph{effective domain} under $\pi^Y$, from which $\ps$ escapes with only negligible probability. More precisely, we denote the $\delta$-blowup of $\mc K$ by $\K_{\delta} = \{x\in \Rd: d(x,\K)\leq \delta\}$.
\begin{lem}\label{lem:eff-domain}
Let $R=(1+h\sigma^{-2})\,\K_{\delta}$ with $\delta=t/d$ and $h=\sigma^{2}c^{2}(d^{-1}\wedge(d^{2}\sigma^{2}-c^{2})^{-1})$, where $c$ is any constant smaller than $d^{2}\sigma^{2}$ and $t\geq2c(c+1)$. Then,
\[
\pi^Y(R^{c})\leq\exp\bpar{c^{2}-\frac{t^{2}}{8c^{2}}}\,.
\]
\end{lem}

\begin{proof}
Using the density formula for $\pi^Y$ in \eqref{eq:pi-y}, 
\begin{align*}
 & \int_{\K}\exp\bpar{-\frac{1}{2\sigma^{2}}\norm x^{2}}\,\D x\cdot\pi^Y(R^{c})\\
 & =\int_{[(1+h\sigma^{-2})\,\K_{\delta}]^{c}}\frac{\int_{\mc K}\exp\bpar{-\frac{1}{2}(\sigma^{-2}+h^{-1})\norm{x-(1+h\sigma^{-2})^{-1}y}^{2}}\,\D x}{(2\pi h)^{d/2}}\,\exp\bpar{-\frac{1}{2(h+\sigma^{2})}\norm y^{2}}\,\D y\\
 & \underset{(i)}{=}\frac{(1+h\sigma^{-2})^{d}}{(2\pi h)^{d/2}}\int_{\mc K_{\delta}^{c}}\int_{\mc K}\exp\bpar{-\frac{\sigma^{-2}+h^{-1}}{2}\norm{x-z}^{2}}\exp\bpar{-\frac{1}{2(h+\sigma^{2})}(1+h\sigma^{-2})^{2}\norm z^{2}}\,\D x\D z\\
 & =\frac{(1+h\sigma^{-2})^{d}}{(2\pi h)^{d/2}}\int_{\mc K_{\delta}^{c}}\int_{\mc K}\exp\bpar{-\frac{\sigma^{-2}+h^{-1}}{2}\norm{x-z}^{2}}\exp\bpar{-\frac{1+h\sigma^{-2}}{2\sigma^{2}}\norm z^{2}}\,\D x\D z\\
 & \underset{(ii)}{\leq}\frac{(1+h\sigma^{-2})^{d}}{(2\pi h)^{d/2}}\int_{\mc K_{\delta}^{c}}\int_{\mc H(z)}\exp\bpar{-\frac{\sigma^{-2}+h^{-1}}{2}\norm{x-z}^{2}}\exp\bpar{-\frac{1+h\sigma^{-2}}{2\sigma^{2}}\norm z^{2}}\,\D x\D z\\
 & =(1+h\sigma^{-2})^{d/2}\int_{\mc K_{\delta}^{c}}\int_{d(z,\mc K)}^{\infty}\sqrt{\frac{\sigma^{-2}+h^{-1}}{2\pi}}\exp\bpar{-\frac{(\sigma^{-2}+h^{-1})y^{2}}{2}}\exp\bpar{-\frac{1+h\sigma^{-2}}{2\sigma^{2}}\norm z^{2}}\,\D y\D z\,,
\end{align*}
where $(i)$ follows from the change of variables $z=(1+h\sigma^{-2})^{-1}y$, and in $(ii)$ $\mc H(z)$ denotes the supporting half-space at $\msf{proj}_{\mc K}(z)$ containing $\mc K$ for given $z\in\de\mc K_{\delta}$, given as
\[
    \mc H(z) = \{x \in \R^d: \inner{\msf{proj}_{\mc K}(z)-z, x- \msf{proj}_{\mc K}(z)} \geq 0 \}\,,
\]
when $z \not\in \mc K$.

We define the one dimensional Gaussian integral
\[
\msf T(s)=\P_{z\sim\mc N\bpar{0,(\sigma^{-2}+h^{-1})^{-1}}}(z\geq s)=\sqrt{\frac{\sigma^{-2}+h^{-1}}{2\pi}}\int_{s}^{\infty}\exp\Bpar{-\frac{(\sigma^{-2}+h^{-1}) y^{2}}{2}}\,\D y\,.
\]
By the co-area formula and integration by parts, for $\mc H^{d-1}$ the $(d-1)$-dimensional Hausdorff measure
\begin{align*}
 & \int_{\mc K_{\delta}^{c}}\int_{d(z,\mc K)}^{\infty}\sqrt{\frac{\sigma^{-2}+h^{-1}}{2\pi}}\exp\Bpar{-\frac{(\sigma^{-2}+h^{-1})y^{2}}{2}}\exp\Bpar{-\frac{1+h\sigma^{-2}}{2\sigma^{2}}\norm z^{2}}\,\D y\D z\\
 & =\int_{\delta}^{\infty}\msf T(s)\int_{\de\mc K_{s}}\exp\Bpar{-\frac{1+h\sigma^{-2}}{2\sigma^{2}}\norm z^{2}}\,\mc H^{d-1}(\D z)\,\D s\\
 & =\Bbrack{\underbrace{\msf T(s)\int_{0}^{s}\int_{\de\mc K_{r}}\exp\Bpar{-\frac{1+h\sigma^{-2}}{2\sigma^{2}}\norm z^{2}}\,\mc H^{d-1}(\D z)\,\D r}_{\eqqcolon\msf I}}_{s=\delta}^{\infty}\\
 & \qquad+\int_{\delta}^{\infty}\sqrt{\frac{\sigma^{-2}+h^{-1}}{2\pi}}\exp\Bpar{-\frac{(\sigma^{-2}+h^{-1})s^{2}}{2}}\int_{0}^{s}\int_{\de\mc K_{r}}\exp\Bpar{-\frac{1+h\sigma^{-2}}{2\sigma^{2}}\norm z^{2}}\,\mc H^{d-1}(\D z)\,\D r \D s\,.
\end{align*}
The double integral above can be bounded as follows:
\begin{align*}
\int_{0}^{s}\int_{\de\mc K_{c}}\exp\bpar{-\frac{1+h\sigma^{-2}}{2\sigma^{2}}\norm z^{2}}\,\mc H^{d-1}(\D z)\,\D c & =\int_{\mc K_{s}\backslash\mc K}\exp\bpar{-\frac{1+h\sigma^{-2}}{2\sigma^{2}}\norm z^{2}}\,\D z\\
 & \underset{(i)}{\le}\int_{(1+s)\,\K}\exp\bpar{-\frac{1+h\sigma^{-2}}{2\sigma^{2}}\norm z^{2}}\,\D z\\
 & =(1+s)^{d}\int_{\mc K}\exp\bpar{-\frac{(1+s)^{2}(1+h\sigma^{-2})}{2\sigma^{2}}\norm z^{2}}\,\D z\\
 & \le(1+s)^{d}\int_{\mc K}\exp\bpar{-\frac{1}{2\sigma^{2}}\norm z^{2}}\,\D z\,.
\end{align*}
where we used $\mc K_{s}\subset(1+s)\mc K$ in $(i)$, which follows from $B_{1}(0)\subset\mc K$. Hence, the double integral is bounded by $(1+s)^{d}\vol(\mc K)$.
Recall a standard tail bound for a Gaussian distribution: 
\[
\msf T(s)\leq\half\exp\Bpar{-\half\bpar{s(\sigma^{-2}+h^{-1})^{1/2}-1}^{2}}\,.
\]
Combining these two bounds, it holds that $\msf I$ vanishes at $s=\infty$. 

Upper bounding $\msf I$ by $0$, we have just derived that
\begin{align*}
\int_{\K}&\exp\bpar{-\frac{1}{2\sigma^{2}}\norm x^{2}}\,\D x\cdot\pi^Y(R^{c}) \\
&\qquad \leq(1+h\sigma^{-2})^{d/2}\int_{\delta}^{\infty}(1+s)^{d}\sqrt{\frac{\sigma^{-2}+h^{-1}}{2\pi}}\exp\bpar{-\frac{(\sigma^{-2}+h^{-1})s^{2}}{2}}\,\D s \cdot \int_{\K}\exp\bpar{-\frac{1}{2\sigma^{2}}\norm z^{2}} \, \D z\,.
\end{align*}
Dividing both sides by the factor $\int_{\K}\exp(-\frac{1}{2\sigma^{2}}\norm x^{2}) \, \D x$, we obtain the following bound,
\begin{align*}
\pi^Y(R^{c}) & \leq(1+h\sigma^{-2})^{d/2}\int_{\delta}^{\infty}(1+s)^{d}\sqrt{\frac{\sigma^{-2}+h^{-1}}{2\pi}}\exp\bpar{-\frac{(\sigma^{-2}+h^{-1})s^{2}}{2}}\,\D s\\
& \leq(1+h\sigma^{-2})^{d/2}\int_{\delta}^{\infty}\exp(sd) \sqrt{\frac{\sigma^{-2}+h^{-1}}{2\pi}}\exp\bpar{-\frac{(\sigma^{-2}+h^{-1})s^{2}}{2}}\,\D s\\
 &\underset{(i)}{\leq} \frac{1}{2}(1+h\sigma^{-2})^{d/2}\exp\bpar{\frac{h'd^{2}}{2}}\exp\Bpar{-\half\bpar{\frac{\delta}{\sqrt{h'}}-d\sqrt{h'}-1}^{2}}\,,
\end{align*}
where in $(i)$ we again use the tail bound for a Gaussian distribution. 
Above, we introduced a new variable $h':=(\sigma^{-2}+h^{-1})^{-1}=\frac{h}{1+h\sigma^{-2}}$. Taking $\delta=t/d$ and $h'=c^{2}/d^{2}$ subject to $t\geq2c(c+1)$, we can make 
\[
\exp\bpar{\frac{h'd^{2}}{2}}\exp\Bpar{-\half\bpar{\frac{\delta}{\sqrt{h'}}-d\sqrt{h'}-1}^{2}}\leq\exp\bpar{\frac{c^{2}}{2}-\frac{t^{2}}{8c^{2}}}\,.
\]
Since $h\leq\sigma^{2}c^{2}/d$, we also have $(1+h\sigma^{-2})^{d/2}\leq\exp(c^{2}/2)$, and the claim follows.
\end{proof}

We can now provide the per-step complexity of the $\ps$ under a warm start.

\begin{lem}\label{lem:num-membership-queries}
Let $\K$ be a convex body in $\Rd$, and $\mu$ an initial distribution $M$-warm with respect to $\pi^{X}$. For any given $n\in\mathbb{N}$ and $\eta\in(0,1)$, set $Z=\frac{9nM}{\eta}(\geq9)$, $h=\sigma^{2}\frac{\log\log Z}{\log Z}(d^{-1}\wedge(d^{2}\sigma^{2}-\frac{\log\log Z}{2\log Z})^{-1})$ and $N=Z(\log Z)^{4}=\Otilde(\frac{nM}{\eta})$. Then, the failure probability of one iteration is at most $\eta/n$. Moreover, the expected number of membership queries needed per iteration is $\O\bpar{M(\log\frac{nM}{\eta})^{4}}$.
\end{lem}

\begin{proof}
For $\mu_{h}:=\mu*\mc N(0,hI_{d})$, the failure probability is $\E_{\mu_{h}}[(1-\ell)^{N}]$. Since $\D\mu/\D\pi^{X}\leq M$ implies $\D\mu_{h}/\D(\pi^{X})_{h}=\D\mu_{h}/\D\pi^{Y}\leq M$, it follows that $\E_{\mu_{h}}[(1-\ell)^{N}]\leq M\,\E_{\pi^{Y}}[(1-\ell)^{N}]$.

We now bound the expectation. For the effective domain $R=(1+h\sigma^{-2})\K_{\delta}$,
\begin{align*}
\int_{\Rd}\underbrace{(1-\ell)^{N}\,\D\pi^{Y}}_{\eqqcolon\msf A} & =\int_{R^{c}}\msf A+\int_{R\cap[\ell\geq N^{-1}\log(3nM/\eta)]}\msf A+\int_{R\cap[\ell<N^{-1} \log(3nM/\eta)]}\msf A\\
 &\underset{(i)}{\leq} \pi^{Y}(R^{c})+\int_{[\ell\geq N^{-1}\log(3nM/\eta)]}\exp(-\ell N)\,\D\pi^{Y}\\
&\qquad+\int_{R\cap[\ell<N^{-1}\log(3nM/\eta)]}\frac{(1+h\sigma^{-2})^{-d/2}\ell(y)}{\int_{\mc K}\exp\bpar{-\frac{1}{2\sigma^{2}}\norm x^{2}}\,\D x}\exp\bpar{-\frac{1}{2(h+\sigma^{2})}\norm y^{2}}\,\D y\\
 & \underset{(ii)}{\leq}  e^{c_{1}^{2}-\frac{t^{2}}{8c_{1}^{2}}}+\frac{\eta}{3nM}+\frac{\log(3nM/\eta)}{N}\,\frac{\int_{R}(1+h\sigma^{-2})^{-d/2}\exp\bpar{-\frac{1}{2(h+\sigma^{2})}\norm y^{2}}\,\D y}{\int_{\mc K}\exp\bpar{-\frac{1}{2\sigma^{2}}\norm x^{2}}\,\D x}\\
 & \underset{(iii)}{\leq} e^{c_{1}^{2}}\exp(-\frac{t^{2}}{8c_{1}^{2}})+\frac{\eta}{3nM}+\frac{\log(3nM/\eta)}{N}\cdot(1+h\sigma^{-2})^{d/2}(1+\delta)^{d}\\
 & \underset{(iv)}{\leq} e^{c_{1}^{2}}\exp(-\frac{t^{2}}{8c_{1}^{2}})+\frac{\eta}{3nM}+\frac{\exp(t+c^{2})}{N}\,\log\frac{3nM}{\eta}\,,
\end{align*}
where in $(i)$ we bounded the $(1-\ell)^N \leq 1$ in the first term, $(1-\ell)^N\leq \exp(-\ell N)$ in the second term, and again $(1-\ell)^N \leq 1$ in the third term, and used the density formula \eqref{eq:density-ell} of $\pi^Y$ in third term. In $(ii)$, the first bound follows from Lemma~\ref{lem:eff-domain}, while the second and third uses the condition on $\ell$ over each domain. $(iii)$ follows from the change of variables. Lastly, $(iv)$ follows from the setup $\delta=t/d$ and $h=\sigma^{2}c^{2}(d^{-1}\wedge(d^{2}\sigma^{2}-c^{2})^{-1})\leq\sigma^2 c^2d^{-1}$ in Lemma~\ref{lem:eff-domain}.

With $c^{2}=\frac{\log\log Z}{2\log Z}$, $t=\sqrt{8}\log\log Z$, and $N=Z(\log Z)^{4}$, the last line is bounded by $\frac{\eta}{nM}$.
Therefore, 
\[
\E_{\mu_{h}}[(1-\ell)^{N}]\leq M\,\E_{\pi^{Y}}[(1-\ell)^{N}]\leq\frac{\eta}{n}\,.
\]

We now bound the expected number of trials per iteration. Let $S$ be the minimum of the threshold $N$ and the number of trials until the first success. Then the expected number of trials per step is bounded by $M\E_{\pi^{Y}}[S]$ due to $\D\mu_{h}/\D\pi^{Y}\leq M$. Thus, 
\begin{align*}
\int_{\Rd}\Bpar{\frac{1}{\ell}\wedge N}\,\D\pi^{Y} & \leq\int_{R}\frac{1}{\ell}\,\D\pi^{Y}+N\pi^{Y}(R^{c})=\frac{\int_{R}(1+h\sigma^{-2})^{-d/2}\exp\bpar{-\frac{1}{2(h+\sigma^{2})}\norm y^{2}}\,\D y}{\int_{\mc K}\exp\bpar{-\frac{1}{2\sigma^{2}}\norm x^{2}}\,\D x}+N\pi^{Y}(R^{c})\\
 & \leq\exp(t+c_{1}^{2})+N\exp(-\Omega(t^{2}))\leq e(\log Z)^{3}+3(\log Z)^{4}=\O\Bpar{\bpar{\log\frac{nM}{\eta}}^{4}}\,.
\end{align*}
Therefore, the expected number of trials per step is $\O\bpar{M(\log\frac{nM}{\eta})^{4}}$, and the claim follows since each trial uses one query to the membership oracle of $\mc K$.
\end{proof}

\subsubsection{Query complexity under a warm start}

We now put together the mixing and per-step guarantees established above.

\begin{prop}\label{prop:ps-gauss-Rq}
For any $\eta,\varepsilon \in (0,1)$, $q\geq 1$, $n\in \mathbb{N}$ defined below and convex body $\K$ given by a well-defined membership oracle, the $\ps$ (Algorithm~\ref{alg:prox-gaussian}) with $h = \sigma^{2}(\frac{1}{d}\wedge\frac{1}{d^{2}\sigma^{2}-1})(\log\frac{9nM}{\eta})^{-1}$, $N = \Otilde(\frac{nM}{\eta})$, and initial distribution $\mu_0^X$ which is $M$-warm with respect to $\pi^X = \mc N(0, \sigma^2 I_d)|_{\K}$ achieves  $\eu R_q(\mu^X_n \mmid \pi^X)\leq \veps$ after $n = \Otilde\bigl(q(d\vee d^2\sigma^2) \log\frac{M}{\eta \veps}\bigr)$ iterations, where $\mu_n^X$ is the law of the $n$-the iterate. 
With probability $1-\eta$, the algorithm iterates $n$ times without failure, using $\Otilde\bpar{qM(d\vee d^{2}\sigma^{2})(\log\nicefrac{1}{\eta\veps})^5}$ expected number of membership queries in total. In particular, the query complexity is $\Otilde(qMd^{2}\sigma^{2}(\log\nicefrac{1}{\eta\veps})^5)$ when $d^{-1}\lesssim \sigma^2$.
\end{prop}

\begin{proof}
By Lemma~\ref{lem:contraction-prox-gauss}, the $\ps$ should iterate
$\O\bpar{q\sigma^{2}h^{-1}\log\frac{\log M}{\veps}}$ times to achieve $\veps$-distance in $\eu R_{q}$. To ensure that the query complexity is bounded, we choose
\[
h=\sigma^{2}\frac{\log\log Z}{\log Z}\bpar{\frac{1}{d}\wedge\frac{1}{d^{2}\sigma^{2}-1}}\quad \text{and}\quad N = \frac{9nM}{\eta} \bpar{\log \frac{9nM}{\eta}}^4\,.
\]
By Lemma~\ref{lem:num-membership-queries}, we need the following total number of membership queries in expectation:
\[
\Otilde\Bpar{qM(d\vee d^{2}\sigma^{2})\bpar{\log\frac{1}{\eta\veps}}^5}\,.
\]
Hence, if $d^{-1}\lesssim\sigma^{2}$, then the query complexity is simply $\Otilde(qMd^{2}\sigma^{2}\log^5\nicefrac{1}{\eta\veps})$.    
\end{proof}

\subsection{R\'enyi-infinity guarantees for truncated Gaussians under warm start}\label{sec:prox-gauss-anystart}

To use the boosting technique, we need the uniform ergodicity of the $\ps$, bounding $\norm{\delta_x P^n- \pi^X}_{\tv}$ uniformly over $\K$, where $P$ is the Markov kernel of the $\ps$ for a truncated Gaussian. To this end, for any $x\in \K$, we bound $\eu R_\infty(\delta_x P^1\mmid\pi^X)$ \emph{uniformly} by $\exp(\text{poly}(D,d))$, so that $\log \log \eu R_\infty(\delta_x P^1\mmid\pi^X)$ does not add more than a polylogarithmic factors to the complexity. In other words, we establish a Gaussian analogue of Lemma~\ref{lem:prox-unif-any-start}.

\begin{lem}\label{lem:prox-gauss-anystart}
For any given $\varepsilon\in(0,1)$, the $\ps$ for a truncated Gaussian $\pi^X$ with variance $h$ and any feasible start $x_{0}\in\mc K$ achieves $\eu R_{\infty}(\mu_{n}^{X}\mmid\pi^{X})\leq \veps$ for $n=\Otilde(h^{-1}\sigma^{2}\log\frac{d+h^{-1}D^{2}}{\veps})$ iterations.
\end{lem}

\begin{proof}
We bound the warmness of $\mu_{1}^{X}$ towards $\pi^{X}$ when $\mu_{0}^{X}=\delta_{x_{0}}$.
One can readily check that  

\[
\mu_{1}^{X}(x)=\ind_{\mc K}(x)\cdot\int\frac{\exp\bpar{-\frac{1+h\sigma^{-2}}{2h}\bnorm{x-\frac{1}{1+h\sigma^{-2}}y}^{2}}\exp\bpar{-\frac{1}{2h}\norm{y-x_{0}}^{2}}}{(2\pi h)^{d/2}\int_{\mc K}\exp\bpar{-\frac{1+h\sigma^{-2}}{2h}\bnorm{z-\frac{1}{1+h\sigma^{-2}}y}^{2}}\,\D z}\,\D y\,,
\]
and we should compare this with 
\[
\pi^X(x) = \frac{\exp\bpar{-\frac{1}{2\sigma^{2}}\norm x^{2}}\cdot\ind_{\mc K}(x)}{\int_{\mc K}\exp\bpar{-\frac{1}{2\sigma^{2}}\norm z^{2}}\,\D z}\,.
\]

For $D=\text{diam}(\mc K)$,
\begin{align*}
\exp\Bpar{-\frac{1+h\sigma^{-2}}{2h}\bnorm{z-\frac{1}{1+h\sigma^{-2}}y}^{2}} & =\exp\bpar{-\frac{1}{2h}\norm z^{2}-\frac{1}{2\sigma^{2}}\norm z^{2}-\frac{1}{2h(1+h\sigma^{-2})}\norm y^{2}+\frac{z^{\T}y}{h}}\\
 & \geq\exp\bpar{-\frac{D^{2}}{2h}}\cdot\exp\bpar{-\frac{1}{2\sigma^{2}}\norm z^{2}-\frac{1}{2h(1+h\sigma^{-2})}\norm y^{2}+\frac{z^{\T}y}{h}}\,.
\end{align*}
As $|z^{\T}y|\leq\half(2\norm z^{2}+\half\norm y^{2})$ due to Young's inequality, 
\begin{align*}
 & \exp\bpar{-\frac{D^{2}}{2h}}\cdot\exp\bpar{-\frac{1}{2\sigma^{2}}\norm z^{2}-\frac{1}{2h(1+h\sigma^{-2})}\norm y^{2}+\frac{z^{\T}y}{h}}\\
\geq & \exp\bpar{-\frac{D^{2}}{2h}}\cdot\exp\bpar{-\frac{1}{2\sigma^{2}}\norm z^{2}-\frac{1}{2h(1+h\sigma^{-2})}\norm y^{2}-\frac{\norm z^{2}}{h}-\frac{\norm y^{2}}{4h}}\\
\geq & \exp\bpar{-\frac{3D^{2}}{2h}}\cdot\exp\bpar{-\frac{1}{2\sigma^{2}}\norm z^{2}-\frac{1}{2h(1+h\sigma^{-2})}\norm y^{2}-\frac{\norm y^{2}}{4h}}\,.
\end{align*}
Hence, 
\begin{align*}
 & \int\frac{\exp\bpar{-\frac{1+h\sigma^{-2}}{2h}\bnorm{x-\frac{1}{1+h\sigma^{-2}}y}^{2}}\exp\bpar{-\frac{1}{2h}\norm{y-x_{0}}^{2}}}{(2\pi h)^{d/2}\int_{\mc K}\exp\bpar{-\frac{1+h\sigma^{-2}}{2h}\bnorm{z-\frac{1}{1+h\sigma^{-2}}y}^{2}}\,\D z}\,\D y\\
 \leq & \underbrace{\frac{\exp\bpar{\frac{3D^{2}}{2h}}}{(2\pi h)^{d/2}}\int\frac{\exp\bpar{-\frac{1+h\sigma^{-2}}{2h}\bnorm{x-\frac{1}{1+h\sigma^{-2}}y}^{2}}\exp\bpar{-\frac{1}{2h}\norm{y-x_{0}}^{2}}\exp\bpar{\frac{1}{2h(1+h\sigma^{-2})}\norm y^{2}+\frac{\norm y^{2}}{4h}}}{\int_{\mc K}\exp\bpar{-\frac{1}{2\sigma^{2}}\norm z^{2}}\,\D z}\,\D y}_{\eqqcolon (\#)}\,.
\end{align*}
The numerator of the integrand can be bounded as follows:
\begin{align*}
    &\exp\bpar{-\frac{1+h\sigma^{-2}}{2h}\bnorm{x-\frac{1}{1+h\sigma^{-2}}y}^{2}}\exp\bpar{-\frac{1}{2h}\norm{y-x_{0}}^{2}}\exp\bpar{\frac{1}{2h(1+h\sigma^{-2})}\norm y^{2}+\frac{\norm y^{2}}{4h}}\\
    =& \exp\Bpar{-\frac{1}{4h}\norm{y-2(x_{0}+x)}^{2}-\frac{1}{2h}\bpar{\norm{x_{0}}^{2}+\norm x^{2}+\frac{h}{\sigma^{2}}\norm x^{2}-2\norm{x_{0}+x}^{2}}}\\
    \leq& \exp\bpar{\frac{D^2}{h}} \exp\bpar{-\frac{1}{4h}\norm{y-2(x_0+x)}^2-\frac{1}{2\sigma^2}\norm{x}^2}\,.
\end{align*}
Putting this bound back to the integral above,
\begin{align*}
    (\#) &\leq \frac{\exp\bpar{\frac{5D^{2}}{2h}}}{(2\pi h)^{d/2}}\frac{\int\exp\bpar{-\frac{1}{4h}\norm{y-2(x_{0}+x)}^{2}}\,\D y}{\int_{\mc K}\exp\bpar{-\frac{1}{2\sigma^{2}}\norm z^{2}}\,\D z}\cdot\exp\bpar{-\frac{1}{2\sigma^{2}}\norm x^{2}}\\
    &= 2^{d/2}\exp\Bpar{\frac{5D^{2}}{h}}\frac{\exp\bpar{-\frac{1}{2\sigma^{2}}\norm x^{2}}}{\int_{\mc K}\exp\bpar{-\frac{1}{2\sigma^{2}}\norm z^{2}}\,\D z} \leq 2^{d/2}\exp\Bpar{\frac{5D^{2}}{h}}\,.
\end{align*}
Therefore, the ratio can be bounded by 
\begin{align*}
 & \frac{\int_{\mc K}\exp\bpar{-\frac{1}{2\sigma^{2}}\norm z^{2}}\,\D z}{\exp\bpar{-\frac{1}{2\sigma^{2}}\norm x^{2}}} \int\frac{\exp\bpar{-\frac{1+h\sigma^{-2}}{2h}\bnorm{x-\frac{1}{1+h\sigma^{-2}}y}^{2}}\exp\bpar{-\frac{1}{2h}\norm{y-x_{0}}^{2}}}{(2\pi h)^{d/2}\int_{\mc K}\exp\bpar{-\frac{1+h\sigma^{-2}}{2h}\bnorm{z-\frac{1}{1+h\sigma^{-2}}y}^{2}}\,\D z}\,\D y\leq2^{d/2}\exp\Bpar{\frac{5D^{2}}{h}}\,,
\end{align*}
so $M=\esssup\frac{\mu_{1}^{X}}{\pi^{X}}\leq2^{d/2}\exp(5h^{-1}D^{2})$, and $\eu R_{q}(\mu_{1}^{X}\mmid\pi^{X})\lesssim\frac{q}{q-1}\log M\leq\frac{q}{q-1}\bpar{d+h^{-1}D^{2}}$. By Lemma~\ref{lem:contraction-prox-gauss}, one can achieve $\eu R_q(\mu_n^X\mmid \pi^X) \leq \veps$ for $n = \Otilde(qh^{-1}\sigma^2 \log\frac{d+h^{-1}D^2}{\veps})$.

Using $2\norm{\cdot}_{\tv}^{2}\leq\KL=\lim_{q\downarrow1}\eu R_{q}$ and $\clsi(\pi^X)\leq \sigma^2$, one can achieve that after $n\gtrsim h^{-1}\sigma^{2}\log\frac{M(d+h^{-1}D^{2})}{\veps}$ iterations,
\[
\sup_{x\in\mc K}\norm{\mu^X_n-\pi^X}_{\tv}\leq\frac{\veps}{M}\,.
\]
By Theorem~\ref{thm:boosting} with $\bnorm{\frac{\D\mu^X_1}{\D\pi^X}-1}_{L_{\infty}}\leq M$, we have $\bnorm{\frac{\D\mu^X_n}{\D\pi^X}-1}_{L_{\infty}}\leq\veps$ and $\eu R_{\infty}(\mu^X_n \mmid\pi)\leq\veps$.
\end{proof}

Using the uniform ergodicity above, we can now obtain a guarantee in $\eu R_\infty$ of the $\ps$ for a truncated Gaussian over $\K$, boosting the metric in Proposition~\ref{prop:ps-gauss-Rq}.

\begin{thm}\label{thm:prox-sampler-truncated-gsn-final}
    For any $\eta,\varepsilon \in (0,1)$, $n\in \mathbb{N}$ defined below, and convex body $\K$ given by a well-defined membership oracle, the $\ps$ (Algorithm~\ref{alg:prox-gaussian}) with $h = \sigma^{2}\bpar{\frac{1}{d}\wedge\frac{1}{d^{2}\sigma^{2}-1}}\bpar{\log\frac{9nM}{\eta}}^{-1}$, $N = \Otilde(\frac{nM}{\eta})$, and initial distribution $\mu_0^X$ that is $M$-warm with respect to $\pi^X$ the truncated Gaussian $\mc N(0, \sigma^2 I_d)|_{\K}$ achieves  $\eu R_\infty(\mu^X_n \mmid \pi^X)\leq \veps$ after $n = \Otilde\bigl((d\vee d^2\sigma^2) \bigl(\log\frac{MD}{\eta \veps}\bigr)^2\bigr)$ iterations, where $\mu_n^X$ is the law of the $n$-the iterate. 
    With probability $1-\eta$, the algorithm iterates $n$ times without failure, using $\Otilde\bpar{M(d\vee d^{2}\sigma^{2})(\log\nicefrac{D}{\eta\veps})^6}$ expected number of membership queries in total. In particular, the query complexity is $\Otilde(Md^{2}\sigma^{2}(\log\nicefrac{D}{\eta\veps})^6)$ when $d^{-1}\lesssim \sigma^2$.
\end{thm}

\begin{proof}
 By Lemma~\ref{lem:prox-gauss-anystart}, the $\ps$ should iterate $n=\Otilde\bigl((d\vee d^2\sigma^2)(\log\frac{MD}{\eta\veps})^2\bigr)$ times to output a sample whose law is $\veps$-close to $\pi^X$ in $\eu R_\infty$. 
 As argued in Proposition~\ref{prop:ps-gauss-Rq}, we can conclude that throughout $n$ outer iterations, the failure probability of the $\ps$ is at most $\eta$ by the union bound, and the total expected number of queries is $\Otilde\bigl(M(d\vee d^2\sigma^2)(\log\frac{D}{\eta\veps})^6\bigr)$.
\end{proof}

\begin{rmk}
    It can be checked that, conditioning on the event that the algorithm has not failed, the distribution of the iterates remains the correct distribution $\mu_n^X$. Thus, in practical terms, the $1-\eta$ probability of failure will not be a significant obstacle, since one can just use the samples of the successful trials without compromising any of the guarantees.
\end{rmk}

\section{Successive proximal scheme: R\'enyi infinity guarantee for uniform sampling}\label{scn:successive-prox}

We put together the ingredients prepared in previous sections, namely the $\ps$ for uniform distributions (Theorem~\ref{thm:ps-unif-Rinfty}) and truncated Gaussian (Theorem~\ref{thm:prox-sampler-truncated-gsn-final}), along with the annealing scheme introduced in \cite{cousins2018gaussian}, and obtain $\sps$ (Algorithm~\ref{alg:sps}).

In \S\ref{sec:sps-results}, we outline $\sps$, together with our main result showing that it can sample a point approximately uniformly from a well-rounded convex body with query complexity $\Otilde(d^3\polylog(1/\veps))$. In \S\ref{sec:sps-proof}, we then provide a proof for the main claim.

\subsection{R\'enyi infinity guarantee with cubic complexity}\label{sec:sps-results}

As mentioned earlier in \S\ref{sec:intro}, using the rounding algorithm in \cite{jia2021reducing}, we can assume that a given convex body $\K \subset \Rd$, presented by a well-defined membership oracle, is \emph{well-rounded}. This means that $\K$ satisfies $\E_{X\sim \K}[\norm X^{2}]\leq C^{2}d$ with a known constant $C>0$, where $X\sim \K$ indicates that $X$ is drawn from the uniform distribution over $\K$.
We state the main result of this paper below.

\begin{thm}\label{thm:main-result}
Assume that a well-rounded convex body $\K$ with $\E_{\mc K}[\norm X^{2}]\leq C^{2}d$ and $0\in \K$ is presented by a well-defined membership oracle.
Let $\pi$ be the uniform distribution over $\mc K$. 
For given $\eta, \veps \in (0,1)$, $\sps$ is a randomized algorithm, which succeeds with probability $1-\eta$ in outputting a sample $X\sim\nu$ such that $\eu R_{\infty}(\nu\mmid\pi)\leq\veps$. Conditioned on its success, $\sps$ uses $\Otilde(C^{2}d^{3}\log^8 \frac{1}{\eta\veps})$ membership queries in expectation.
\end{thm}

\paragraph{Preliminaries.}
We first collect a series of observations that simplify our setup and arguments. For $\bar{X}:=\E_{\mc K}[X]$, by Jensen's inequality, 
\[
\norm{\bar{X}}\leq\E_{\mc K}[\norm X] \leq\sqrt{\E_{\mc K}[\norm X^{2}]}\leq C\sqrt{d}\,.
\]
Also, for the covariance matrix $\Sigma_{\K}$ of the uniform distribution over $\K$,
\[
\tr(\Sigma_{\K}) = \E_{\K}[\norm{X-\bar{X}}^2] \leq \E_{\K}[\norm{X}^2] \leq C^2d\,.
\]
Lastly, it is known from \cite[Theorem 5.17]{lovasz2007geometry} that $\P_{\mc K}\bpar{\norm{X-\bar{X}}\geq t\cdot \bpar{\tr(\Sigma_{\K})}^{1/2}} \leq\exp(-t+1)$. Hence,
\[
\P_{\mc K}\bpar{\norm{X-\bar{X}}\geq t\cdot C\sqrt{d}} \leq\exp(-t+1)\,.
\]
Therefore, we can actually work with a `truncated' convex body instead of the full convex body $\K$:
\begin{prop}
    There exists a constant $L=C\log\frac{3e}{\veps}$ such that the volume of $\bK:=\mc K\cap B_{L\sqrt{d}}(0)$ is at least $(1-\veps/3)\vol(\mc K)$.
\end{prop}

\paragraph{Sketch of $\sps$.}
Let $\mu_{i}$ be a truncated Gaussian $\mc N(0,\sigma_{i}^{2}I_{d})|_{\mc K}$, and $\bmu_{i}$ be an approximate measure close to $\mu_{i}$ produced by the algorithm. Then, $\sps$ consists of four phases:
\begin{itemize}
\item Phase I $(\sigma^{2}=1/d$)
\begin{itemize}
\item Initial distribution: Uniform measure over $B_{1}(0)$ denoted by
$\mu_{0}=\bmu_{0}$.
\item Target distribution: $\mc N(0,\sigma^{2}I_{d})|_{\bK}=\mu_{1}$. 
\end{itemize}
\item Phase II ($1/d\leq\sigma^{2}\leq1$)
\begin{itemize}
\item Run $\ps$ with initial dist. $\bmu_{i}$ (not $\mu_{i}$) and
target dist. $\mu_{i+1}$.
\item Update $\sigma_{i+1}^{2}=\sigma_{i}^{2}\bpar{1+\frac{1}{d}}$.
\end{itemize}
\item Phase III ($1\leq\sigma^{2}\leq L^{2}d$)
\begin{itemize}
\item Run $\ps$ with initial dist. $\bmu_{i}$ (not $\mu_{i}$) and
target dist. $\mu_{i+1}$.
\item Update $\sigma_{i+1}^{2}=\sigma_{i}^{2}\bpar{1+\frac{\sigma_{i}^{2}}{L^{2}d}}$.
\end{itemize}
\item Phase IV ($\sigma^{2}=L^{2}d$)
\begin{itemize}
\item Run $\ps$ with initial distribution $\mc N(0,\sigma^{2}I)|_{\bK}$
and target distribution $\pi_{\mc{\bar{K}}}$, the uniform distribution over
$\bK$.
\end{itemize}
\end{itemize}

\begin{rmk}
    Phases I-III can be viewed as ``preprocessing steps'' whose purpose is to generate a warm start for Phase IV. We also note that, while $\eu R_{\infty}(\nu \mmid \pi) \leq \varepsilon$, we do not have the (slightly stronger) property that, for $\varepsilon < 1$,
    \[
        \esssup_{\pi} \Big|\frac{\D \nu}{\D \pi} - 1\Bigr| \lesssim \varepsilon\,.
    \]
    This is because the truncation of $\pi=\pi_{\K}$ to $\bar{\pi}=\pi_{\bar{\K}}$ causes there to be $\mc A= \K\backslash \bar{\K}$ where $\bar{\pi}(\mc A) = 0 < \pi(\mc A)$. Since $\nu \ll \bar{\pi}$, we cannot establish a lower bound on $\frac{\D \nu}{\D \pi} - 1$ better than $0$. On the other hand, Corollary~\ref{cor:lp-main} shows that we can bound
    \[
        \Bnorm{\frac{\D \nu}{\D \pi} - 1}_{L^p(\pi)}^p \leq \varepsilon \,.
    \]
\end{rmk}

\begin{cor}\label{cor:lp-main}
    Under the same assumptions as Theorem~\ref{thm:main-result}, $\sps$ succeeds with probability $1-\eta$ in outputting a sample $X \sim \nu$ such that $\norm{\frac{\D \nu}{\D \pi} - 1}_{L^p(\pi)}^p \leq \varepsilon$, so long as $\varepsilon < 1$, $p \geq 1$.
    Conditioned on its success, $\sps$ uses $\Otilde(C^{2}d^{3}\log^8 \frac{1}{\eta\veps})$ membership queries in expectation.
\end{cor}

\begin{proof}
    As an intermediate step in the proof of Theorem~\ref{thm:main-result}, we obtain that $\esssup_{\bar \pi} \bigl\vert \frac{\D \nu}{\D \bar \pi} -1\bigr\vert \leq \frac{\varepsilon}{3}$, so that, replacing $\bK$ with $\K$,
    \begin{align*}
        \E_\pi \Bbrack{\Bigl\vert \frac{\D \nu}{\D \pi} -1\Bigr\vert^p} 
        &\leq \frac{\vol(\K \backslash \bK)}{\vol(\K)} + \frac{\vol(\bK)}{\vol(\K)}\, \Bigl\vert \frac{\D \nu}{\D \bar{\pi}}\, \frac{\vol(\K)}{\vol(\bK)}-1\Bigr\vert^p  \\
        &\leq \frac{\varepsilon}{3} + \bigl(1-\frac{\varepsilon}{3}\bigr) \Bigl(\frac{2\varepsilon}{3(1-\frac{\varepsilon}{3})}\Bigr)^p \leq \varepsilon\,,
    \end{align*}
    so long as $\varepsilon < 1$ and $p \geq 1$. 
\end{proof}

\subsection{Proof details}\label{sec:sps-proof}
The following subsections state a series of lemmas, each one of which proves a guarantee on each phase of the algorithm. 

\paragraph{Failure probability.}
For a target failure probability $\eta$ of the entire algorithm, we can achieve this by setting the failure probability at each phase to be on the order of
\begin{equation}
    \hat \eta := \Otilde \Bigl(\frac{\eta}{C^2 d \log^2 \frac{1}{\varepsilon}}\Bigr)\,,\label{eq:hat-eta-def}
\end{equation}
since the total number of inner phases is $\Otilde(C^2 d \log^2 \frac{1}{\varepsilon})$.
This will be sufficient for the failure probability to be at most $\eta$, where we apply a union bound over all the phases. Since the dependence on $\hat \eta$ of the $\ps$ is polylogarithmic, this will not impact the resulting oracle complexity bounds by more than polylogarithmic factors.

\paragraph{Structure of lemmas.}
We provide a lemma for each outer phase, where each lemma establishes three facts: (i) the number of updates (to $\sigma$), (ii) a quantitative guarantee on the warmness on each update, and (iii) the final query-complexity bound given by Theorem~\ref{thm:prox-sampler-truncated-gsn-final}.

Briefly, this theorem says that, starting from an $M$-warm distribution towards a truncated Gaussian $\mc N(0, \sigma^2 I)|_{\bK}$, the oracle complexity to achieve $\varepsilon$ error in $\eu R_\infty$, with success probability $1-\hat \eta$, is bounded by a quantity of asymptotic order
\[
    \Otilde\bigl(Md^{2}\sigma^{2}\log^6\frac{L}{\varepsilon \hat \eta}\bigr)\,.
\]
Since $\sigma^2 > d^{-1}$ at all times, it suffices to choose $h^{-1} = \widetilde{\Theta}(d^2\log \frac{M \log L}{\varepsilon \hat \eta})$ and the number of $\ps$ iterations to be $n = \Otilde(d^2 \sigma^2 \log^2 \frac{L}{\varepsilon \hat \eta})$. 

Except for the last phase, we take $\varepsilon = \log 2$, in which case the query complexity is simply
\begin{equation}\label{eq:query_complexity}
    \Otilde\bigl(Md^{2}\sigma^{2}\log^6\frac{L}{\hat \eta}\bigr)\,,
\end{equation}
with $h^{-1} = \widetilde{\Theta}(d^2\log \frac{M \log L}{\hat \eta})$ and number of $\ps$ iterations being $n = \Otilde(d^2 \sigma^2 \log^2 \frac{ML}{\hat \eta})$. Note that $\log 2$ error in $\eu R_\infty$ implies that the law of the resultant sample is at least $2$-warm with respect to the target.

\subsubsection{Phase I}

\begin{lem}[Phase I]\label{lem:phase-i-bounds}
    With probability at least $1-\hat \eta$, initial distribution $\mu_0 = \msf{Unif}(B_1(0))$, and target distribution $\mu_1 = \mc{N}(0, \frac{1}{d} I_d) \vert_{\bK}$, \emph{Phase I} outputs a sample with law $\bmu_1$ satisfying 
    \[
        \esssup_{\mu_1} \Bigl|\frac{\D \bmu_1}{\D \mu_1} - 1\Bigr| \leq 1\,.
    \]
    Its oracle complexity is
    \[
        \Otilde\bigl(d^{3/2} \log^6 \frac{L}{\hat \eta}\bigr)\,, 
    \]
    using $\ps$ with step size $h^{-1} = \widetilde{\Theta}(d^2 \log \frac{\log L}{\hat \eta})$ and $\Otilde(d \log^2 \frac{L}{\hat \eta})$ iterations.
\end{lem}
\begin{proof}
\textbf{\# Inner phases:} The number of inner phases is clearly $1$.

\textbf{Warmness:}
The initial measure $\mu_0$ has warmness given by
\begin{align*}
    \frac{\D \mu_0(x)}{\D \mu_1(x)} &= \frac{1/\vol(B_1(0))}{\exp(-\frac{d}2\norm{x}^2)/\int_{\bK} \exp(-\frac{d}2\norm{z}^2)\,\D z}
    \\
    &\leq \frac{(2\pi/d)^{d/2}\exp(\frac{d}2)}{\vol(B_1)} \cdot \underbrace{\frac{1}{(2\pi/d)^{d/2}} \int_{\bK} \exp\bpar{-\frac{d}2\norm{z}^2}\,\D z}_{\leq 1}\\
    &\underset{(i)}{\leq} (2/d)^{d/2}\exp\bigl(\frac{d}2\bigr) \Gamma\bpar{\frac{d}2+1}
    \underset{(ii)}{\leq} \gamma\sqrt{d}\exp\bpar{-\frac{d}{2}\log d + \frac{d}{2} + \frac{d}{2} \log d - \frac{d}{2}}\leq \gamma\sqrt{d}\,,
\end{align*}
where in $(i)$ we used $\vol(B_1) = \pi^{d/2} / \Gamma(\textfrac{d}2+1)$, and $(ii)$ follows from Stirling's approximation that
\[
    \Gamma\bigl(\frac{d}{2} + 1\bigr) \leq \gamma e d^{\frac{d}{2}+1/2} (2 e)^{-\frac{d}{2}}\quad\text{for some universal constant }\gamma>0\,.
\]

\textbf{Final complexity:} It follows from substituting $M = \O(\sqrt{d})$ and $\sigma^2 = d^{-1}$ into~\eqref{eq:query_complexity}.
\end{proof}

\subsubsection{Phase II}
\begin{lem}[Phase II]\label{lem:phase-ii-bounds}
With probability at least $1-\hat \eta i_*$ for $i_* = \Otilde(d)$, initial distribution $\bmu_1$ (given by from Lemma~\ref{lem:phase-i-bounds}), and target distribution $\mu_{i_*} = \mc{N}(0, I_d) \vert_{\bK}$, \emph{Phase II} outputs a sample with law $\bmu_{i_*}$ satisfying 
\[
    \esssup_{\mu_{i_*}} \Bigl|\frac{\D \bmu_{i_*}}{\D \mu_{i_*}} - 1\Bigr| \leq 1\,,
\]
Its oracle complexity is
\[
    \Otilde\bigl(d^{3} \log^6 \frac{L}{\hat \eta}\bigr)\,, 
\]
using $\ps$. The step size scheme for each inner phase is presented in the proof.
\end{lem}

\begin{proof}

Note that we slightly modify the scheme, which is to take $\sigma_{i+1}^2 = \min(1, \sigma_i^2 (1+\frac{1}{d}))$.

\textbf{\# Inner phases:}
As $(1+d^{-1})^{d} \geq 2$ for any $d\geq 1$, it takes at most $d$ iterations to double $\sigma$. Thus, since we need to double $\sigma$ on the order of $\mc{O}(\log d)$ many times, the number of inner phases within Phase II is at most $i_*=\Otilde(d)$.

\textbf{Warmness:} By the construction of the algorithm, the following bound holds for all $x \in \bK$:
\begin{align*}
\frac{\D\mu_{i}(x)}{\D\mu_{i+1}(x)} & =\frac{\exp\bpar{-\frac{1}{2\sigma_{i}^{2}}\norm x^{2}}}{\exp\bpar{-\frac{1}{2\sigma_{i+1}^{2}}\norm x^{2}}}\,\frac{\int_{\bK}\exp\bpar{-\frac{1}{2\sigma_{i+1}^{2}}\norm x^{2}}\,\D x}{\int_{\bK}\exp\bpar{-\frac{1}{2\sigma_{i}^{2}}\norm x^{2}}\,\D x}\leq\frac{\int_{\bK}\exp\bpar{-\frac{1}{2\sigma_{i+1}^{2}}\norm x^{2}}\,\D x}{\int_{\bK}\exp\bpar{-\frac{1}{2\sigma_{i}^{2}}\norm x^{2}}\,\D x}\\
 & \leq\frac{\bpar{1+\frac{1}{d}}^{d/2}\int_{(1+\frac{1}{d})^{-1/2}\bK}\exp\bpar{-\frac{1}{2\sigma_{i}^{2}}\norm x^{2}}\,\D x}{\int_{\bK}\exp\bpar{-\frac{1}{2\sigma_{i}^{2}}\norm x^{2}}\,\D x}\leq\sqrt{e}\,.
\end{align*}
Furthermore, if at each previous phase, we have
\begin{equation}\label{eq:warmness-p2}
    \esssup_{\mu_i} \Bigl|\frac{\D \bar{\mu_i}}{\D \mu_i} - 1\Bigr| \leq 1\,,
\end{equation}
then this implies warmness of constant order between $\bmu_i$ and $\mu_{i+1}$:
\[
    \esssup_{\mu_{i+1}}\Bigl|\frac{\D \bar{\mu_i}}{\D\mu_{i+1}}\Bigr| \leq 2\sqrt{e}\,.
\]

\textbf{Final complexity:} For any fixed $i$, the complexity in each inner phase is given by substituting $M \leq 2\sqrt{e}$ and $\sigma^2_i \leq 1$ into~\eqref{eq:query_complexity}, to obtain a complexity of
\[
    \Otilde\bigl(d^2 \log^6 \frac{L}{\hat \eta}\bigr)\,.
\]
Here, we run the $\ps$ with $h_i^{-1} = \widetilde{\Theta}(d^2 \log \frac{\log L}{\hat \eta})$ and $n_i = \Otilde(d^2 \sigma^2 \log^2 \frac{L}{\hat \eta})$, to get $\eu R_\infty$ bounded by $\log 2$. Implicitly, this verifies the condition \eqref{eq:warmness-p2}.

Secondly, as $i_* = \Otilde(d)$, the total query complexity is bounded by the product of the worst-case complexity in each inner phase with the total number of inner phases, which is given by
\[
    \Otilde\bigl(d^3 \log^6 \frac{L}{\hat \eta}\bigr)\,.\qedhere
\]
\end{proof}

\subsubsection{Phase III}
\begin{lem}[Phase III]\label{lem:phase-iii-bounds}
With probability at least $1-\hat \eta j_*$ for $j_* = \Otilde(d)$, initial distribution $\bmu_{i_*}$ (given by from Lemma~\ref{lem:phase-ii-bounds}), and target distribution $\mu_{i_*+j_*} = \mc{N}(0, L^2 dI_d) \vert_{\bK}$, \emph{Phase II} outputs a sample with law $\bmu_{i_*+j_*}$ satisfying 
\[
    \esssup_{\mu_{i_*+j_*}} \Bigl|\frac{\D \bmu_{i_*+j_*}}{\D \mu_{i_*+j_*}} - 1\Bigr| \leq 1\,,
\]
Its oracle complexity is
\[
    \Otilde\bigl(L^2d^{3} \log^6 \frac{1}{\hat \eta}\bigr)\,, 
\]
using the $\ps$. The step size scheme for each inner phase is presented in the proof.
\end{lem}

\begin{proof}
In this phase, we will first perform the analysis over each doubling and then aggregate over the doublings.

\textbf{\# Inner phases:} We first partition $[1, L^2d]$ by a sequence of doubling parts, where the terminal $\sigma^2$ is at least double of the initial $\sigma^2$ in each doubling part. Clearly, the number of doubling parts is $\log_2(L^2d)$.

Let $\sigma^2$ be an initial variance of a given doubling part. Since we have 
\[
    \Bigl(1+\frac{\sigma^2}{L^2 d}\Bigr)^{\frac{L^2 d}{\sigma^2}}\geq 2\quad \text{for all }d\geq 1 \text{ and } \sigma \leq L\sqrt{d}\,,
\]
the number of inner phases within the doubling part is at most $L^2d/\sigma^2\leq L^2d$. Therefore, the total number of inner phases during Phase III is $\Otilde(L^2d)$.

\textbf{Warmness:} Let $j \geq i_*$. By
the construction of the algorithm, for all $x \in \bK$, 
\[
\frac{\D\mu_{j}(x)}{\D\mu_{j+1}(x)} =\frac{\exp\bpar{-\frac{1}{2\sigma_{j}^{2}}\norm x^{2}}}{\exp\bpar{-\frac{1}{2\sigma_{j+1}^{2}}\norm x^{2}}}\,\frac{\int_{\bK}\exp\bpar{-\frac{1}{2\sigma_{j+1}^{2}}\norm x^{2}}\,\D x}{\int_{\bK}\exp\bpar{-\frac{1}{2\sigma_{j}^{2}}\norm x^{2}}\,\D x}\leq\frac{\int_{\bK}\exp\bpar{-\frac{1}{2\sigma_{j+1}^{2}}\norm x^{2}}\,\D x}{\int_{\bK}\exp\bpar{-\frac{1}{2\sigma_{j}^{2}}\norm x^{2}}\,\D x}\,.
\]
As $\norm{x}^2 \leq L^2 d$ on $\bK$, we have
\[
\exp\bpar{-\frac{1}{2\sigma_{j+1}^{2}}\norm x^{2}} =\exp\bpar{-\frac{1}{2\sigma_{j}^{2}}\norm x^{2}}\cdot\exp\bpar{\frac{1}{2(L^{2}d+\sigma_{j}^{2})}\norm x^{2}}\leq \sqrt{e} \exp\bpar{-\frac{1}{2\sigma_{j}^{2}}\norm x^{2}}\,.
\]
As a result,
\[
    \frac{\D\mu_{j}(x)}{\D\mu_{j+1}(x)} \leq \frac{\int_{\bK}\exp\bpar{-\frac{1}{2\sigma_{j+1}^{2}}\norm x^{2}}\,\D x}{\int_{\bK}\exp\bpar{-\frac{1}{2\sigma_{j}^{2}}\norm x^{2}}\,\D x} \leq \frac{\sqrt{e}\int_{\bK}\exp\bpar{-\frac{1}{2\sigma_{j}^{2}}\norm x^{2}}\,\D x}{\int_{\bK}\exp\bpar{-\frac{1}{2\sigma_{j}^{2}}\norm x^{2}}\,\D x}=\sqrt{e}\,.
\]
Furthermore, if at each previous phase, we have
\begin{align}\label{eq:warmness-p3}
    \esssup_{\mu_{j}}\Bigl|\frac{\D \bar{\mu_j}}{\D\mu_{j}} - 1\Bigr| \leq 1\,,
\end{align}
then this implies warmness of constant order
\[
    \esssup_{\mu_{j+1}}\Bigl|\frac{\D \bar{\mu_j}}{\D\mu_{j+1}}\Bigr| \leq 2\sqrt{e}\,.
\]

\textbf{Final complexity:}
We first bound the total query complexity of one doubling part. Let $\sigma^2$ be an initial variance of a given doubling part. For any intermediate variance $\sigma_j^2$ within the part, the query complexity of the $\ps$ with $h^{-1} = \widetilde{\Theta}(d^2\log^2\frac{L}\eta)$, iteration number $n = \Otilde(d^2\sigma_j^2 \log^2 \frac{L}{\hat\eta})$, and warmness $M\geq 2\sqrt{e}$ can be obtained from \eqref{eq:query_complexity} as 
\[
\Otilde\bpar{d^2\sigma_j^2\log^6\frac{L}{\hat\eta}}\,.
\]
This achieves $\log 2$-accuracy in $\eu R_\infty$ in each update, which also implicitly verifies~\eqref{eq:warmness-p3}.
Since within the doubling part, $\sigma_j^2 \leq 4\sigma^2$ (where $\sigma^2$ is the initial variance within the doubling part) and the number of inner phases is at most $L^2d/\sigma^2$, the total query complexity during Phase III, aggregating over all the doubling parts, is
\[
\log_2(L^2 d) \cdot \frac{L^2d}{\sigma^2}\cdot \Otilde\Bpar{d^2(4\sigma^2)\log^6\frac{L}{\hat\eta}}
= \Otilde\bpar{L^2d^3\log^6\frac{1}{\hat\eta}}\,,
\]
achieving $\log 2$-accuracy in $\eu R_\infty$ for the law of the final sample.
\end{proof}

\subsubsection{Phase IV}

\begin{lem}[Phase IV]\label{lem:phase-iv-bounds}
With probability at least $1-\hat \eta$, initial distribution $\bmu_{i_*+j_*}$ (given by from Lemma~\ref{lem:phase-iii-bounds}), and target distribution $\bar{\pi}\propto \ind_{\bK}$, \emph{Phase IV} outputs a sample with law $\nu$ satisfying 
\[
    \esssup_{\bar \pi} \Bigl|\frac{\D \nu}{\D \bar \pi} - 1\Bigr| \leq 
    \frac{\varepsilon}3\,,
\]
Its oracle complexity is
\[
    \Otilde\bigl(L^2 d^{3} \log^6 \frac{1}{\hat \eta\veps}\bigr)\,, 
\]
using the $\ps$ with $h^{-1} = \widetilde{\Theta}(d^2 \log \frac{ \log (L/\varepsilon)}{\hat \eta})$ and $\Otilde(L^2 d^3 \log^2 \frac{1}{\hat \eta \varepsilon})$ iterations.
\end{lem}

\begin{proof}
\textbf{\# Inner phases:} There is only a single update in this phase.

\textbf{Warmness:}
    Let $\bar{\pi}$ be the uniform distribution over $\bK$.
First, bounding $\exp(-\frac{1}{2L^{2}d}\norm x^{2})\leq 1$ and, since $\bK \subseteq B_{L\sqrt{d}}(0)$,
\[
    \int_{\bK}\exp\bpar{-\frac{1}{2L^{2}d}\norm x^{2}}\,\D x \geq \vol(\bK) \exp(-\frac{L^2d}{2L^2 d}) = \vol(\bK)/\sqrt{e}\,.
\]
Then, we can deduce
\[
\frac{\D \mu_{i_* + j_*}(x)}{\D \bar\pi(x)} = \frac{\D \mu_{i_* + j_*}(x)}{1/\vol(\bK)}  =\vol(\bK)\cdot\frac{\exp\bpar{-\frac{1}{2L^{2}d}\norm x^{2}}}{\int_{\bK}\exp\bpar{-\frac{1}{2L^{2}d}\norm x^{2}}\,\D x}\leq\sqrt{e}\,.
\]
As before, since  we have
\[
    \esssup_{\mu_{i_* + j_*}}\Bigl|\frac{\D \bmu_{i_* + j_*}}{\D\mu_{i_* + j_*}} - 1\Bigr| \leq 1\,,
\]
then this implies warmness of constant order
\[
    \esssup_{\mu_{i_* + j_*}}\Abs{\frac{\D \bmu_{i_* + j_*}}{\D\bar\pi}} \leq 2\sqrt{e}\,.
\]

\textbf{Final complexity:} Here, using Theorem~\ref{thm:ps-unif-Rinfty} (instead of the Gaussian results) with $M = 2\sqrt{e}$ and $D = 2L\sqrt{d}$, we obtain the query complexity of
\[
    \Otilde\bigl(L^2 d^3 \log^6 \frac{1}{\hat \eta \varepsilon}\bigr)\,.\qedhere
\]
\end{proof}

We now put together these four lemmas to prove Theorem~\ref{thm:main-result}.
\begin{proof}[Proof of Theorem~\ref{thm:main-result}]
    Given the start $\mu_0$, we first apply Lemma~\ref{lem:phase-i-bounds} to produce $\bmu_1$. Then apply Lemma~\ref{lem:phase-ii-bounds} to produce $\bmu_{i_*}$, Lemma~\ref{lem:phase-iii-bounds} to produce $\bmu_{i_* + j_*}$, and finally Lemma~\ref{lem:phase-iv-bounds} to produce $\nu$. Their complexity is dominated by that of Lemma~\ref{lem:phase-iv-bounds}, which is
    \[
        \Otilde\bigl(L^2 d^3 \log^6 \frac{1}{\hat \eta \varepsilon}\bigr)
        =\Otilde\bpar{C^2 d^3 \log^8 \frac{1}{\hat \eta \veps}}\,,
    \]
    with the guarantee $\eu R_\infty(\nu \mmid \bar{\pi}) \leq \frac{\varepsilon}{3}$.
    
    From our choice of $L$, it follows that $\eu R_\infty(\bar\pi \mmid \pi) \leq \log(1+\frac{2\veps}3)$ due to
    \[
        \sup_{x\in\mc K}\frac{\bar{\pi}(x)}{\pi(x)}=\frac{\vol(\mc K)}{\vol(\bK)}\leq1+\frac{2\veps}3\,.
    \]
    Therefore,
    \begin{align*}
        \eu R_{\infty}(\nu \mmid \pi) 
        &= \esssup_{\pi} \log \frac{\D \nu}{\D \pi} 
        \leq \eu R_\infty (\nu \mmid \bar{\pi})  + \eu R_\infty (\bar{\pi}\mmid \pi)\\
        &\leq \frac{\veps}3 + \log\bpar{1 + \frac{2\veps}3}\leq \veps\,.
    \end{align*}

    Secondly, we need to quantity the probability of successfully generating a sample. This is given, by a union bound, as $1-\hat \eta (2+i_* + j_*)$. With $\hat \eta$ chosen according to~\eqref{eq:hat-eta-def}, $\sps$ has at least $1-\eta$ probability of success. Putting all these together concludes the proof.
\end{proof}

\section{Conclusion}
We have presented $\sps$, an algorithm which achieves $\varepsilon$-closeness in $\eu R_\infty$ with $\O(d^3 \polylog \frac{1}{\varepsilon})$ query complexity. Here we note some possible future directions for research.
\begin{itemize}
    \item In this work, we showed that it is possible to extend the guarantees for $\ino$ for Gaussians of the form $\mc N(0, \sigma^2 I_d)$. Naturally, one asks whether it is possible to sample from general distributions of the type $e^{-f}|_{\mc K}$, when $f$ has some nice properties? The applications of this problem are innumerable, and the techniques required are far from trivial. It is likely that one would need to devise a different annealing scheme in order to control the initial warmness parameter, or find a new strategy which bypasses the need for annealing entirely.

    \item This work assumed the membership oracle model for $\mc K$, but it is interesting to consider alternative oracle models. For instance, one can use more features of the geometry, by assuming access to a self-concordant barrier or a Riemannian metric on $\mc K$. Alternatively, one can assume access to a stronger oracle such as a separation oracle, which returns not just a binary response to if $x \in \mc K$, but when $x \not\in \mc K$ also gives a hyperplane separating $x$ and $\mc K$. It is as yet unknown if such a model can improve the query complexity, even for uniform sampling. It is also possible that a different annealing strategy would be beneficial in this context.
\end{itemize}

\section*{Acknowledgements}
We thank Sinho Chewi for helpful feedback.
YK was supported in part by NSF awards CCF-2007443 and CCF-2134105.
MSZ was supported by NSERC through the CGS-D program. 

\bibliographystyle{alpha}
\bibliography{main}

\newcommand{\etalchar}[1]{$^{#1}$}
\begin{thebibliography}{MHFH{\etalchar{+}}23}

\bibitem[AC21]{ahn2021efficient}
Kwangjun Ahn and Sinho Chewi.
\newblock Efficient constrained sampling via the mirror-{L}angevin algorithm.
\newblock {\em Advances in Neural Information Processing Systems (NeurIPS)}, 34:28405--28418, 2021.

\bibitem[AC23]{altschuler2023faster}
Jason~M Altschuler and Sinho Chewi.
\newblock Faster high-accuracy log-concave sampling via algorithmic warm starts.
\newblock In {\em Symposium on Foundations of Computer Science (FOCS)}, pages 2169--2176. IEEE, 2023.

\bibitem[AK91]{applegate1991sampling}
David Applegate and Ravi Kannan.
\newblock Sampling and integration of near log-concave functions.
\newblock In {\em Symposium on Theory of Computing (STOC)}, pages 156--163, 1991.

\bibitem[BCJ{\etalchar{+}}19]{bingham2019pyro}
Eli Bingham, Jonathan~P. Chen, Martin Jankowiak, Fritz Obermeyer, Neeraj Pradhan, Theofanis Karaletsos, Rohit Singh, Paul~A. Szerlip, Paul Horsfall, and Noah~D. Goodman.
\newblock Pyro: {D}eep universal probabilistic programming.
\newblock {\em The Journal of Machine Learning Research (JMLR)}, 20:28:1--28:6, 2019.

\bibitem[BDMP17]{pmlr-v65-brosse17a}
Nicolas Brosse, Alain Durmus, \'Eric Moulines, and Marcelo Pereyra.
\newblock Sampling from a log-concave distribution with compact support with proximal {L}angevin {M}onte {C}arlo.
\newblock In {\em Conference on Learning Theory (COLT)}, volume~65, pages 319--342. PMLR, 2017.

\bibitem[BEL18]{bubeck2018sampling}
S{\'e}bastien Bubeck, Ronen Eldan, and Joseph Lehec.
\newblock Sampling from a log-concave distribution with projected {L}angevin {M}onte {C}arlo.
\newblock {\em Discrete \& Computational Geometry (DCG)}, 59:757--783, 2018.

\bibitem[BGL14]{bakry2014analysis}
Dominique Bakry, Ivan Gentil, and Michel Ledoux.
\newblock {\em Analysis and geometry of {M}arkov diffusion operators}, volume 103.
\newblock Springer, 2014.

\bibitem[BST14]{bassily2014private}
Raef Bassily, Adam Smith, and Abhradeep Thakurta.
\newblock Private empirical risk minimization: {E}fficient algorithms and tight error bounds.
\newblock In {\em Symposium on Foundations of Computer Science (FOCS)}, pages 464--473. IEEE, 2014.

\bibitem[CB18]{cheng2018convergence}
Xiang Cheng and Peter Bartlett.
\newblock Convergence of {L}angevin {MCMC} in {KL}-divergence.
\newblock In {\em Algorithmic Learning Theory (ALT)}, pages 186--211. PMLR, 2018.

\bibitem[CCSW22]{chen2022improved}
Yongxin Chen, Sinho Chewi, Adil Salim, and Andre Wibisono.
\newblock Improved analysis for a proximal algorithm for sampling.
\newblock In {\em Conference on Learning Theory (COLT)}, pages 2984--3014. PMLR, 2022.

\bibitem[CDWY18]{chen2018fast}
Yuansi Chen, Raaz Dwivedi, Martin~J Wainwright, and Bin Yu.
\newblock Fast {MCMC} sampling algorithms on polytopes.
\newblock {\em The Journal of Machine Learning Research (JMLR)}, 19(1):2146--2231, 2018.

\bibitem[CEL{\etalchar{+}}22]{chewi2021analysis}
Sinho Chewi, Murat~A Erdogdu, Mufan Li, Ruoqi Shen, and Shunshi Zhang.
\newblock Analysis of {L}angevin {Monte Carlo} from {P}oincar{\'e} to {Log-Sobolev}.
\newblock In {\em Conference on Learning Theory (COLT)}, pages 1--2. PMLR, 2022.

\bibitem[CG23]{chen2023does}
Yuansi Chen and Khashayar Gatmiry.
\newblock When does {M}etropolized {H}amiltonian {M}onte {C}arlo provably outperform {M}etropolis-adjusted {L}angevin algorithm?
\newblock {\em arXiv preprint arXiv:2304.04724}, 2023.

\bibitem[CLA{\etalchar{+}}21]{chewi2021optimal}
Sinho Chewi, Chen Lu, Kwangjun Ahn, Xiang Cheng, Thibaut Le~Gouic, and Philippe Rigollet.
\newblock Optimal dimension dependence of the {M}etropolis-adjusted {L}angevin algorithm.
\newblock In {\em Conference on Learning Theory (COLT)}, pages 1260--1300. PMLR, 2021.

\bibitem[CLL19]{cao2019exponential}
Yu~Cao, Jianfeng Lu, and Yulong Lu.
\newblock Exponential decay of {R}{\'e}nyi divergence under {F}okker--{P}lanck equations.
\newblock {\em Journal of Statistical Physics}, 176:1172--1184, 2019.

\bibitem[CV16]{cousins2016practical}
Ben Cousins and Santosh~S Vempala.
\newblock A practical volume algorithm.
\newblock {\em Mathematical Programming Computation}, 8(2):133--160, 2016.

\bibitem[CV18]{cousins2018gaussian}
Ben Cousins and Santosh~S Vempala.
\newblock Gaussian {C}ooling and {$O^{*}(n^3)$} algorithms for volume and {G}aussian volume.
\newblock {\em SIAM Journal on Computing (SICOMP)}, 47(3):1237--1273, 2018.

\bibitem[CZS22]{cheng2022theory}
Xiang Cheng, Jingzhao Zhang, and Suvrit Sra.
\newblock Theory and algorithms for diffusion processes on {R}iemannian manifolds.
\newblock {\em arXiv preprint arXiv:2204.13665}, 2022.

\bibitem[Dal17]{dalalyan2017further}
Arnak Dalalyan.
\newblock Further and stronger analogy between sampling and optimization: {Langevin Monte Carlo} and gradient descent.
\newblock In {\em Conference on Learning Theory (COLT)}, pages 678--689. PMLR, 2017.

\bibitem[DCM19]{delahaye2019simulated}
Daniel Delahaye, Supatcha Chaimatanan, and Marcel Mongeau.
\newblock Simulated annealing: {F}rom basics to applications.
\newblock {\em Handbook of metaheuristics}, pages 1--35, 2019.

\bibitem[DCWY18]{dwivedi2018log}
Raaz Dwivedi, Yuansi Chen, Martin~J Wainwright, and Bin Yu.
\newblock Log-concave sampling: {Metropolis-Hastings} algorithms are fast!
\newblock In {\em Conference on Learning Theory (COLT)}, pages 793--797. PMLR, 2018.

\bibitem[DFK91]{dyer1991random}
Martin Dyer, Alan Frieze, and Ravi Kannan.
\newblock A random polynomial-time algorithm for approximating the volume of convex bodies.
\newblock {\em Journal of the ACM (JACM)}, 38(1):1--17, 1991.

\bibitem[DMLM03]{del2003contraction}
Pierre Del~Moral, Michel Ledoux, and Laurent Miclo.
\newblock On contraction properties of {M}arkov kernels.
\newblock {\em Probability Theory and Related Fields}, 126(3):395--420, 2003.

\bibitem[DMM19]{durmus2019analysis}
Alain Durmus, Szymon Majewski, and Blazej Miasojedow.
\newblock Analysis of {Langevin Monte Carlo} via convex optimization.
\newblock {\em The Journal of Machine Learning Research (JMLR)}, 20(1):2666--2711, 2019.

\bibitem[Doo53]{doob1953stochastic}
Joseph~L Doob.
\newblock {\em Stochastic Processes}.
\newblock John Wiley and Sons,, 1953.

\bibitem[DT12]{dalalyan2012sparse}
Arnak~S Dalalyan and Alexandre~B Tsybakov.
\newblock Sparse regression learning by aggregation and {L}angevin {M}onte-{C}arlo.
\newblock {\em Journal of Computer and System Sciences (JCSS)}, 78(5):1423--1443, 2012.

\bibitem[EHZ22]{erdogdu2022convergence}
Murat~A Erdogdu, Rasa Hosseinzadeh, and Shunshi Zhang.
\newblock Convergence of {L}angevin {M}onte {C}arlo in chi-squared and {R}{\'e}nyi divergence.
\newblock In {\em International Conference on Artificial Intelligence and Statistics (AISTATS)}, pages 8151--8175. PMLR, 2022.

\bibitem[GC11]{girolami2011riemann}
Mark Girolami and Ben Calderhead.
\newblock Riemann manifold {Langevin and Hamiltonian Monte Carlo} methods.
\newblock {\em Journal of the Royal Statistical Society: Series B (Statistical Methodology)}, 73(2):123--214, 2011.

\bibitem[GHZ22]{gurbuzbalaban2022penalized}
Mert G{\"u}rb{\"u}zbalaban, Yuanhan Hu, and Lingjiong Zhu.
\newblock Penalized {L}angevin and {H}amiltonian {M}onte {C}arlo algorithms for constrained sampling.
\newblock {\em arXiv preprint arXiv:2212.00570}, 2022.

\bibitem[GKV23]{gatmiry2023sampling}
Khashayar Gatmiry, Jonathan Kelner, and Santosh~S Vempala.
\newblock Sampling with barriers: {F}aster mixing via {L}ewis weights.
\newblock {\em arXiv preprint arXiv:2303.00480}, 2023.

\bibitem[GLL22]{gopi2022private}
Sivakanth Gopi, Yin~Tat Lee, and Daogao Liu.
\newblock Private convex optimization via exponential mechanism.
\newblock In {\em Conference on Learning Theory (COLT)}, pages 1948--1989. PMLR, 2022.

\bibitem[GLL{\etalchar{+}}23]{gopi2023private}
Sivakanth Gopi, Yin~Tat Lee, Daogao Liu, Ruoqi Shen, and Kevin Tian.
\newblock Private convex optimization in general norms.
\newblock In {\em Symposium on Discrete Algorithms (SODA)}, pages 5068--5089. SIAM, 2023.

\bibitem[GLS88]{grotschel2012geometric}
Martin Gr{\"o}tschel, L{\'a}szl{\'o} Lov{\'a}sz, and Alexander Schrijver.
\newblock {\em Geometric algorithms and combinatorial optimization}, volume~2.
\newblock Springer, 1988.

\bibitem[GT20]{ganesh2020faster}
Arun Ganesh and Kunal Talwar.
\newblock Faster differentially private samplers via {R}{\'e}nyi divergence analysis of discretized {L}angevin {MCMC}.
\newblock {\em Advances in Neural Information Processing Systems (NeurIPS)}, 33:7222--7233, 2020.

\bibitem[HCT{\etalchar{+}}17]{haraldsdottir2017chrr}
Hulda~S Haraldsd{\'o}ttir, Ben Cousins, Ines Thiele, Ronan~MT Fleming, and Santosh~S Vempala.
\newblock Chrr: coordinate hit-and-run with rounding for uniform sampling of constraint-based models.
\newblock {\em Bioinformatics}, 33(11):1741--1743, 2017.

\bibitem[HJJ03]{henderson2003theory}
Darrall Henderson, Sheldon~H Jacobson, and Alan~W Johnson.
\newblock The theory and practice of simulated annealing.
\newblock {\em Handbook of metaheuristics}, pages 287--319, 2003.

\bibitem[HT10]{hardt2010geometry}
Moritz Hardt and Kunal Talwar.
\newblock On the geometry of differential privacy.
\newblock In {\em Symposium on Theory of Computing (STOC)}, pages 705--714, 2010.

\bibitem[Jia21]{jiang2021mirror}
Qijia Jiang.
\newblock Mirror {L}angevin {M}onte {C}arlo: the case under isoperimetry.
\newblock {\em Advances in Neural Information Processing Systems (NeurIPS)}, 34:715--725, 2021.

\bibitem[JLLV21]{jia2021reducing}
He~Jia, Aditi Laddha, Yin~Tat Lee, and Santosh~S Vempala.
\newblock Reducing isotropy and volume to {KLS}: an {$O^*(n^3 \psi^2)$} volume algorithm.
\newblock In {\em Symposium on Theory of Computing (STOC)}, pages 961--974, 2021.

\bibitem[KGJV83]{kirkpatrick1983optimization}
Scott Kirkpatrick, C~Daniel Gelatt~Jr, and Mario~P Vecchi.
\newblock Optimization by simulated annealing.
\newblock {\em Science}, 220(4598):671--680, 1983.

\bibitem[Kla23]{klartag2023logarithmic}
Bo\'az Klartag.
\newblock Logarithmic bounds for isoperimetry and slices of convex sets.
\newblock {\em Ars Inveniendi Analytica}, 2023.

\bibitem[KLM06]{kannan2006blocking}
Ravi Kannan, L{\'a}szl{\'o} Lov{\'a}sz, and Ravi Montenegro.
\newblock Blocking conductance and mixing in random walks.
\newblock {\em Combinatorics, Probability and Computing}, 15(4):541--570, 2006.

\bibitem[KLS95]{kannan1995isoperimetric}
Ravi Kannan, L{\'a}szl{\'o} Lov{\'a}sz, and Mikl{\'o}s Simonovits.
\newblock Isoperimetric problems for convex bodies and a localization lemma.
\newblock {\em Discrete \& Computational Geometry}, 13(3):541--559, 1995.

\bibitem[KLS97]{kannan1997random}
Ravi Kannan, L{\'a}szl{\'o} Lov{\'a}sz, and Mikl{\'o}s Simonovits.
\newblock Random walks and an {$O^*(n^5)$} volume algorithm for convex bodies.
\newblock {\em Random Structures \& Algorithms (RS\&A)}, 11(1):1--50, 1997.

\bibitem[KLSV22]{kook2022sampling}
Yunbum Kook, Yin~Tat Lee, Ruoqi Shen, and Santosh~S Vempala.
\newblock Sampling with {R}iemannian {H}amiltonian {M}onte {C}arlo in a constrained space.
\newblock In {\em Advances in Neural Information Processing Systems (NeurIPS)}, volume~35, pages 31684--31696, 2022.

\bibitem[KLSV23]{kook2022condition}
Yunbum Kook, Yin~Tat Lee, Ruoqi Shen, and Santosh~S Vempala.
\newblock Condition-number-independent convergence rate of {R}iemannian {H}amiltonian {M}onte {C}arlo with numerical integrators.
\newblock In {\em Conference on Learning Theory (COLT)}, volume 195, pages 4504--4569. PMLR, 2023.

\bibitem[KN12]{kannan2012random}
Ravindran Kannan and Hariharan Narayanan.
\newblock Random walks on polytopes and an affine interior point method for linear programming.
\newblock {\em Mathematics of Operations Research}, 37(1):1--20, 2012.

\bibitem[KV06]{kalai2006simulated}
Adam~Tauman Kalai and Santosh~S Vempala.
\newblock Simulated annealing for convex optimization.
\newblock {\em Mathematics of Operations Research}, 31(2):253--266, 2006.

\bibitem[KV24]{kook2024gaussian}
Yunbum Kook and Santosh~S Vempala.
\newblock Gaussian cooling and {D}ikin walks: {T}he interior-point method for logconcave sampling.
\newblock In {\em Conference on Learning Theory (COLT)}, pages 3137--3240. PMLR, 2024.

\bibitem[KVZ24]{kook2024inout}
Yunbum Kook, Santosh~S Vempala, and Matthew~S Zhang.
\newblock In-and-{O}ut: {A}lgorithmic diffusion for sampling convex bodies.
\newblock {\em arXiv preprint arXiv:2405.01425}, 2024.

\bibitem[LDM15]{lindsten2015uniform}
Fredrik Lindsten, Randal Douc, and Eric Moulines.
\newblock Uniform ergodicity of the particle {G}ibbs sampler.
\newblock {\em Scandinavian Journal of Statistics}, 42(3):775--797, 2015.

\bibitem[LE23]{li2020riemannian}
Mufan Li and Murat~A Erdogdu.
\newblock Riemannian {L}angevin algorithm for solving semidefinite programs.
\newblock {\em Bernoulli}, 29(4):3093--3113, 2023.

\bibitem[Leh23]{lehec2023langevin}
Joseph Lehec.
\newblock The {L}angevin {M}onte {C}arlo algorithm in the non-smooth log-concave case.
\newblock {\em The Annals of Applied Probability}, 33(6A):4858--4874, 2023.

\bibitem[Lov91]{lovasz1991compute}
L{\'a}szl{\'o} Lov{\'a}sz.
\newblock {\em How to compute the volume?}
\newblock DIMACS, Center for Discrete Mathematics and Theoretical Computer Science, 1991.

\bibitem[Lov99]{lovasz1999hit}
L{\'a}szl{\'o} Lov{\'a}sz.
\newblock Hit-and-run mixes fast.
\newblock {\em Mathematical Programming}, 86:443--461, 1999.

\bibitem[LS90]{lovasz1990mixing}
L{\'a}szl{\'o} Lov{\'a}sz and Mikl{\'o}s Simonovits.
\newblock The mixing rate of {M}arkov chains, an isoperimetric inequality, and computing the volume.
\newblock In {\em Symposium on Foundations of Computer Science (FOCS)}, pages 346--354. IEEE, 1990.

\bibitem[LS93]{lovasz1993random}
L{\'a}szl{\'o} Lov{\'a}sz and Mikl{\'o}s Simonovits.
\newblock Random walks in a convex body and an improved volume algorithm.
\newblock {\em Random Structures \& Algorithms (RS\&A)}, 4(4):359--412, 1993.

\bibitem[LST21]{lee2021structured}
Yin~Tat Lee, Ruoqi Shen, and Kevin Tian.
\newblock Structured logconcave sampling with a restricted {G}aussian oracle.
\newblock In {\em Conference on Learning Theory (COLT)}, pages 2993--3050. PMLR, 2021.

\bibitem[LTVW22]{li2022mirror}
Ruilin Li, Molei Tao, Santosh~S Vempala, and Andre Wibisono.
\newblock The mirror {L}angevin algorithm converges with vanishing bias.
\newblock In {\em International Conference on Algorithmic Learning Theory (ALT)}, pages 718--742. PMLR, 2022.

\bibitem[LV06a]{lovasz2006fast}
L{\'a}szl{\'o} Lov{\'a}sz and Santosh~S Vempala.
\newblock Fast algorithms for logconcave functions: {S}ampling, rounding, integration and optimization.
\newblock In {\em Symposium on Foundations of Computer Science (FOCS)}, pages 57--68. IEEE, 2006.

\bibitem[LV06b]{lovasz2006hit}
L{\'a}szl{\'o} Lov{\'a}sz and Santosh~S Vempala.
\newblock Hit-and-run from a corner.
\newblock {\em SIAM Journal on Computing (SICOMP)}, 35(4):985--1005, 2006.

\bibitem[LV06c]{lovasz2006simulated}
L{\'a}szl{\'o} Lov{\'a}sz and Santosh~S Vempala.
\newblock Simulated annealing in convex bodies and an {${O}^{*}(n^{4})$} volume algorithm.
\newblock {\em Journal of Computer and System Sciences (JCSS)}, 72(2):392--417, 2006.

\bibitem[LV07]{lovasz2007geometry}
L{\'a}szl{\'o} Lov{\'a}sz and Santosh~S Vempala.
\newblock The geometry of logconcave functions and sampling algorithms.
\newblock {\em Random Structures \& Algorithms (RS\&A)}, 30(3):307--358, 2007.

\bibitem[LV18]{lee2018convergence}
Yin~Tat Lee and Santosh~S Vempala.
\newblock Convergence rate of {Riemannian Hamiltonian Monte Carlo} and faster polytope volume computation.
\newblock In {\em Symposium on Theory of Computing (STOC)}, pages 1115--1121, 2018.

\bibitem[Mar06]{markov1906extension}
Andrey~Andreyevich Markov.
\newblock Extension of the law of large numbers to dependent quantities.
\newblock {\em Izv. Fiz.-Matem. Obsch. Kazan Univ.(2nd Ser)}, 15(1):135--156, 1906.

\bibitem[MCC{\etalchar{+}}21]{ma2021there}
Yi-An Ma, Niladri~S Chatterji, Xiang Cheng, Nicolas Flammarion, Peter~L Bartlett, and Michael~I Jordan.
\newblock Is there an analog of {N}esterov acceleration for gradient-based {MCMC}?
\newblock {\em Bernoulli}, 27(3), 2021.

\bibitem[MFWB22]{mou2022improved}
Wenlong Mou, Nicolas Flammarion, Martin~J Wainwright, and Peter~L Bartlett.
\newblock Improved bounds for discretization of {L}angevin diffusions: {N}ear-optimal rates without convexity.
\newblock {\em Bernoulli}, 28(3):1577--1601, 2022.

\bibitem[MHFH{\etalchar{+}}23]{mousavi2023towards}
Alireza Mousavi-Hosseini, Tyler~K Farghly, Ye~He, Krishna Balasubramanian, and Murat~A Erdogdu.
\newblock Towards a complete analysis of {L}angevin {M}onte {C}arlo: {B}eyond {P}oincar{\'e} inequality.
\newblock In {\em Conference on Learning Theory (COLT)}, pages 1--35. PMLR, 2023.

\bibitem[Mir17]{mironov2017renyi}
Ilya Mironov.
\newblock R{\'e}nyi differential privacy.
\newblock In {\em Computer Security Foundations Symposium (CSF)}, pages 263--275. IEEE, 2017.

\bibitem[MT07]{mcsherry2007mechanism}
Frank McSherry and Kunal Talwar.
\newblock Mechanism design via differential privacy.
\newblock In {\em Foundations of Computer Science (FOCS)}, pages 94--103. IEEE, 2007.

\bibitem[MT12]{meyn2012markov}
Sean~P Meyn and Richard~L Tweedie.
\newblock {\em Markov chains and stochastic stability}.
\newblock Springer Science \& Business Media, 2012.

\bibitem[MV22]{mangoubi2022sampling}
Oren Mangoubi and Nisheeth Vishnoi.
\newblock Sampling from log-concave distributions with infinity-distance guarantees.
\newblock In {\em Advances in Neural Information Processing Systems (NeurIPS)}, volume~35, pages 12633--12646, 2022.

\bibitem[Rud11]{Rudolf2011ExplicitEB}
Daniel Rudolf.
\newblock Explicit error bounds for {M}arkov chain {M}onte {C}arlo.
\newblock {\em Dissertationes Mathematicae}, 485:1--93, 2011.

\bibitem[Smi84]{smith1984efficient}
Robert~L Smith.
\newblock Efficient {M}onte {C}arlo procedures for generating points uniformly distributed over bounded regions.
\newblock {\em Operations Research}, 32(6):1296--1308, 1984.

\bibitem[{Sta}20]{stan}
{Stan Development Team}.
\newblock {RStan}: the {R} interface to {Stan}, 2020.
\newblock R package version 2.21.2.

\bibitem[SWW24]{srinivasan2023fast}
Vishwak Srinivasan, Andre Wibisono, and Ashia Wilson.
\newblock Fast sampling from constrained spaces using the {M}etropolis-adjusted {M}irror {L}angevin algorithm.
\newblock In {\em Conference on Learning Theory (COLT)}, pages 4593--4635. PMLR, 2024.

\bibitem[VW19]{vempala2019rapid}
Santosh~S Vempala and Andre Wibisono.
\newblock Rapid convergence of the unadjusted {L}angevin algorithm: {I}soperimetry suffices.
\newblock {\em Advances in Neural Information Processing Systems (NeurIPS)}, 32, 2019.

\bibitem[WL17]{wang2017geometric}
Zheng Wang and Cong Ling.
\newblock On the geometric ergodicity of {M}etropolis-{H}astings algorithms for lattice {G}aussian sampling.
\newblock {\em Transactions on Information Theory}, 64(2):738--751, 2017.

\bibitem[YK41]{yosida1941operator}
K{\^o}saku Yosida and Shizuo Kakutani.
\newblock Operator-theoretical treatment of {M}arkoff's process and mean ergodic theorem.
\newblock {\em Annals of Mathematics}, pages 188--228, 1941.

\bibitem[ZCL{\etalchar{+}}23]{zhang2023improved}
Matthew~S Zhang, Sinho Chewi, Mufan Li, Krishna Balasubramanian, and Murat~A Erdogdu.
\newblock Improved discretization analysis for underdamped {L}angevin {M}onte {C}arlo.
\newblock In {\em Conference on Learning Theory (COLT)}, pages 36--71. PMLR, 2023.

\bibitem[ZPFP20]{zhang2020wasserstein}
Kelvin~Shuangjian Zhang, Gabriel Peyr{\'e}, Jalal Fadili, and Marcelo Pereyra.
\newblock Wasserstein control of mirror {L}angevin {M}onte {C}arlo.
\newblock In {\em Conference on Learning Theory (COLT)}, pages 3814--3841. PMLR, 2020.

\end{thebibliography}

\end{document}